\documentclass[aps,superscriptaddress,showpacs, nofootinbib, pre,10pt,twocolumn]{revtex4-2}
\usepackage{amsmath,amsthm,amssymb,hyperref,enumitem,physics,float,xcolor,mdframed , graphicx}
\usepackage{dsfont}
\usepackage[bottom]{footmisc}
\usepackage{dsfont}

\newcommand{\Z}{\mathbb{Z}}
\newcommand{\E}{\mathbb{E}}

\theoremstyle{definition}
\newtheorem{definition}{Definition}

\newtheorem{theorem}{Theorem}
\newtheorem{conjecture}{Conjecture}
\newtheorem{lemma}{Lemma}

\newtheorem{corollary}[theorem]{Corollary}

\usepackage{algorithm}
\usepackage{algpseudocode}

\begin{document}

\title{Stoquasticity is not enough: towards a sharper diagnostic for Quantum Monte Carlo simulability}
\author{Arman Babakhani}
\affiliation{Information Sciences Institute, University of Southern California, Marina del Rey, California 90292, USA}
\affiliation{Department of Physics and Astronomy and Center for Quantum Information Science \& Technology, University of Southern California, Los Angeles, California 90089, USA}
\author{Armen Karakashian}
\affiliation{Computer Science and Mathematics, Rutgers University, New Brunswick, New Jersey, USA}

\begin{abstract}
\noindent Quantum Monte Carlo (QMC) methods are powerful tools for simulating quantum many-body systems, yet their applicability is limited by the infamous sign problem. We approach this challenge through the lens of Vanishing Geometric Phases (VGP) \cite{Hen_2021}, introducing it as a `geometric' criterion for diagnosing QMC simulability. We characterize the class of VGP Hamiltonians, and analyze the complexity of recognizing this class, identifying both hard and efficiently identifiable cases.
We further highlight the practical advantage of the VGP criterion by exhibiting specific Hamiltonians that are readily identified as sign-problem-free through VGP, yet whose stoquasticity is difficult to ascertain. These examples underscore the efficiency and sharpness of VGP as a diagnostic tool compared to stoquasticity-based heuristics.
Beyond classification, we propose a family of VGP-inspired diagnostics that serve as quantitative indicators of sign problem severity. While exact evaluation of these quantities is generically intractable, we demonstrate their mathematical power in performing scaling analysis for the average sign under unitary transformations.
Our results provide both a conceptual foundation and practical tools for understanding and mitigating the sign problem.

\end{abstract}
\maketitle
\tableofcontents

\section{Introduction}

Quantum Monte Carlo (QMC) methods are among the most powerful and widely used numerical techniques for studying quantum many-body systems, particularly in lattice models of interacting spins and fermions. These methods allow for the evaluation of thermal and ground-state properties by stochastically sampling configurations according to a weight determined by the particular method of decomposition. However, in many physically relevant scenarios, especially in geometrically frustrated quantum magnets and fermionic systems, the QMC approach encounters a severe limitation: the sign problem. This limitation prevents accurate QMC simulations in the low temperature (high $\beta$) and high particle number (large $N$), rendering the study of ground state and thermodynamic properties for systems with sign problem intractable \cite{Loh_1990,Troyer_2005,Henelius_2000}.

QMC sign problem is typically understood in terms of stoquasticity: Hamiltonians that are stoquastic are sign problem free. However, the inverse statement is not true. Though, there are not many sign-problem free Hamiltonians for which there does not exist a simple one or two-body Clifford rotations to transform them into stoquastic ones: almost all well-known physical Hamiltonians that are sign problem free can be mapped into a stoquastic Hamiltonian via one or two-body Clifford transformations \cite{Bravyi_2007,Marvian_2019}.
 
In spin-$1/2$ systems, which are foundational in quantum magnetism and play a central role in quantum simulation platforms, the sign problem often poses a critical barrier to tractable simulations. This difficulty is especially pronounced in frustrated quantum antiferromagnets, such as the triangular or Kagome lattice Heisenberg models, where competing exchange interactions generate intricate phase interference patterns that manifest as negative or complex-valued weights in standard quantum Monte Carlo (QMC) formulations \cite{Alet_2016}. These sign fluctuations, ubiquitous in various QMC methods, result in severe sampling inefficiencies that grow exponentially with system size or inverse temperature. Consequently, simulations of such models remain limited to relatively small lattices, even on modern supercomputers \cite{Henelius_2000}. Even though some progress has been made in mitigating the severity of the QMC sign in certain cases \cite{DEmidio_2020,Wessel_2017,Sato_2021,Hangleiter_2020,Nakamura_1998}, understanding the origin and structure of the sign problem in these models thus remains an active and fundamental pursuit.

In fermionic models like the Fermi-Hubbard model away from half-filling, the sign problem is generically present and similarly limits simulation capabilities\cite{Pan_2022}. There have been attempts to address the origin of the sign problem in such systems as well, under specific QMC framework such as path-integral, Majorana and determinant QMC methods\cite{Iazzi_2014,Li_2015}. There are various ways to map fermionic systems into spin-$1/2$ systems \cite{Fradkin_1989}, however, the sign problem shows its head irrespective of the specific mapping used to describe the fermionic Hamiltonians.

In this paper, some of our results on the complexity of determining whether a Hamiltonian is in the Vanishing Geometric Phase (VGP) class applies to generic systems satisfying the Permutation Matrix Representation (PMR) formalism. However, all of our specific Hamiltonian examples will focus on spin-$1/2$ systems. As mentioned above, since there are specific mappings from fermionic systems to spin-$1/2$ systems, like Jordan-Wigner transformations, our analysis and results could be pertinent to QMC sign problem for fermionic systems.

\section{outline}

In \cite{Hen_2021}, the VGP framework was introduced, offering a broader lens than the conventional notion of stoquasticity. Traditionally, much of the research on mitigating the sign problem has centered on finding unitary or Clifford transformations that render a Hamiltonian stoquastic: an approach that can be both computationally intensive and inconclusive. While the VGP criterion does not directly provide curing transformations for sign-problematic Hamiltonians, it offers a critical advantage: it can aid in identifying the type of transformation that would make no impact on curing the sign problem, thereby avoiding fruitless searches over unitary spaces where the sign problem is irreducible.

Furthermore, in this work, we argue that VGP not only enriches our conceptual understanding of QMC simulability but, in certain instances, provides a more practical and computationally efficient diagnostic than checking for stoquasticity. We also demonstrate how the simulability of a Hamiltonian can often be deduced by analyzing the \emph{weights of fundamental cycles} in its computational state graph: a graph-theoretic criterion that, in many cases, is significantly easier to verify than searching for a basis in which the Hamiltonian becomes stoquastic.

In \cite{Hen_2021}, it was shown that any Hamiltonian satisfying the VGP criterion can be mapped to a stoquastic Hamiltonian via a purely diagonal phase rotation
\begin{align}
    H_{\text{VGP}} = \Phi H_{\text{stoq.}} \Phi^{\dagger} \, ,
    \label{eq:VGP_to_stoq}
\end{align}
where $\Phi = \text{diag}(e^{i\phi_1}, e^{i\phi_2}, \ldots, e^{i\phi_{\mathds{D}}})$ and $\mathds{D}$ denotes the Hilbert space dimension. While this structural connection was established, its practical implications, particularly in relation to computational complexity and practical uses, have remained underexplored.

In this work, we aim to bridge that gap. We begin by reviewing the Permutation Matrix Representation (PMR) and that of PMRQMC in sections~\ref{sec:pmr} and ~\ref{sec:pmrqmc_rev}, respectively. We also introduce the use of fundamental cycle weights from the computational state graph as a diagnostic tool for QMC simulability in section~\ref{sec:compstate_fundcyc}.

In section~\ref{sec:measures}, we introduce concrete metrics for determining when a Hamiltonian satisfies the VGP criterion. We also formalize several results concerning the complexity of determining whether a given Hamiltonian is VGP, particularly for physically relevant models. We further identify a class of unitary transformations that leave these metrics invariant, thereby establishing that certain transformations—such as those belonging to the Pauli group—cannot resolve the sign problem when present.

In section~\ref{sec:VGPscaling_signseverity}, we address how our metric introduced in section~\ref{sec:measures} scale as a function of system size $N$ and inverse temperature $\beta$. Using these scaling arguments, we show that there could be systems for even though the sign problem is present, the average sign of a QMC simulation remains constant as a function of system size and do not plague a QMC calculation for large $N$ calculations. However, we prove that no such behavior is possible as a function of $\beta$. Even though such findings have been addressed in previous work \cite{Mondaini_2023,Iglovikov_2015}, we present a rigorous analysis through the VGP criterion that highlights the distinction between scaling of average sign with $N$ and $\beta$.

In section~\ref{sec:VGPorStoq}, we characterize a class of sparse Hamiltonians that are readily identifiable as VGP. We also present explicit models within this class for which finding a unitary transformation that renders the Hamiltonian stoquastic is demonstrably difficult, highlighting the practical utility of the VGP framework in such settings.

Finally, in section~\ref{sec:avg_sign_stats}, we examine the mathematical properties of the VGP measures introduced in section~\ref{sec:measures}, and provide a rigorous analysis demonstrating the futility of attempting sign problem mitigation through random searches over unitary or Clifford spaces. We prove that random unitary or Clifford transformations are overwhelmingly likely to exacerbate the sign problem, even when applied to Hamiltonians that were originally sign-problem-free. This reinforces a deeper complexity-theoretic intuition: blindly searching for curing transformations is unlikely to succeed, as the general problem of determining the ground state energy of a local Hamiltonian is QMA-hard. While these insights have been widely appreciated on an intuitive level within the QMC community, to the best of our knowledge, no prior work has offered a formal mathematical treatment of these phenomena.

\section{Permutation Matrix Representation}
\label{sec:pmr}
We consider an arbitrary many-body system whose Hamiltonian $H$ is given as input. We first cast $H$ in permutation matrix representation (PMR), i.e., as the sum
\begin{align}
    \label{eq:PMR}
    H =\sum_{j=0}^M D_j P_j =D_0 + \sum_{j=1}^M D_j P_j \,,
\end{align} 
where $\{ P_j\}$ is a set of $M+1$ distinct generalized permutation matrices~\cite{Gupta_2020}. Without loss of generality, the off-diagonal terms of the Hamiltonian can be written as linear combination of $D_j P_j$ where $D_j$ is a diagonal matrix and $P_j$ is a  permutation matrix with no fixed points (equivalently, no nonzero diagonal elements). The diagonal term, on the other hand can be written as $D_0 P_0$, where $P_0$ is the identity matrix, the only (trivial) permutation that has a fixed point. Of course, the choice of the diagonal and permutation matrices are basis dependent. In this work, we will refer to the basis in which the operators $\{D_j\}$ are diagonal as the computational basis and denote its states by $\{ |z\rangle \}$.

The $\{D_j P_j \}$ off-diagonal operators (in the computational basis) give the system its  `quantum dimension'.  Each term $D_j P_j $ obeys 
$D_j P_j \ket{z} = d_{z'}^{(j)} | z' \rangle$ where $d_{z'}^{(j)}$ is a possibly complex-valued coefficient and $|z'\rangle \neq |z\rangle$ is a basis state\footnote{While the above formulation may appear restrictive, it can be shown that any finite-dimensional matrix can be written in the form of Eq.~\eqref{eq:PMR} easily and efficiently~\cite{Ezzell_2025}}. 

We also note the matrix elements of the diagonal operators as follows
\begin{equation}
d_{z_j}^{(i_j)} = \langle z_j \vert D_{i_j} \vert z_j \rangle \,. 
\label{eq:dz}
\end{equation}
The various $\{|z_i\rangle\}$ states are the states obtained from the action of the ordered $P_j$ operators in the product $S_{{\bf{i}}_q}$ on $|z_0\rangle$, then on $|z_1\rangle$, and so forth. For example, for $S_{{\bf{i}}_q}=P_{i_q} \ldots P_{i_2}P_{i_1}$, we obtain $|z_0\rangle=|z\rangle, P_{i_1}|z_0\rangle=|z_1\rangle, P_{i_2}|z_1\rangle=|z_2\rangle$, etc. The proper indexing of the states $|z_j\rangle$ along the path is \hbox{$|z_{(i_1,i_2,\ldots,i_j)}\rangle$} to indicate that the state at the $j$-th step depends on all $P_{i_1}\ldots P_{i_j}$. The sequence of basis states $\{|z_i\rangle \}$ may be viewed as a closed `walk' on the computational state graph of $H$, where every matrix element $H_{ij}$ corresponds to the weight of an edge between the two basis states node $i$ and node $j$. We will rigorously describe this in section \ref{sec:compstate_fundcyc}.
The PMR formalism, thus, provides a firm ground to describe the expectation values of quantum operators, as well as, statistical quantities like the partition function as the action of diagonal and permutation matrices on the computational basis states.

\section{Computational state graphs and fundamental cycles}
\label{sec:compstate_fundcyc}
\subsection{Computational state graphs}
\begin{definition}
Let $H$ be a Hamiltonian acting on a finite-dimensional Hilbert space with a fixed computational basis $\{ \ket{z} \}$. The \emph{computational state graph} associated with $H$, denoted $G_H = (V, E)$, is a directed graph defined as follows:

\begin{enumerate}
    \item The vertex set $V$ consists of the computational basis states, i.e., $V = \{ \ket{z} \}$.
    \item The edge set $E$ contains an edge between $\ket{z}$ and $\ket{z'}$ if and only if the corresponding matrix element $H_{zz'} = \bra{z} H \ket{z'}$ is nonzero and $z \neq z'$. This implies that each edge weight $e_{z z'}$ is the complex conjugate of $e_{z' z}$, ensuring Hermiticity of the Hamiltonian.
\end{enumerate}

That is, $G_H$ encodes the off-diagonal structure of $H$ in the computational basis, with edges representing nonzero transition amplitudes between basis states.
\end{definition}

Using the PMR decomposition, given by Eq. (\ref{eq:PMR}), and Eq. (\ref{eq:dz}), we can infer that the edge weight $e_{z \, z_i}$ corresponds to a matrix element $d^{(i)}_{z_i} = \bra{z_i} D_i P_i \ket{z}$, such that $P_i \ket{z} = \ket{z_i}$. By adopting the notation $\ket{z_{(i_1, i_2 ,\ldots , i_q)}} \equiv P_{i_q} \ldots P_{i_2} P_{i_1} \ket{z}$, we can write the weight of the edge $e_{z_{(i_1, i_2 ,\ldots , i_{q-1})} z_{(i_1, i_2 ,\ldots , i_{q-1} , i_{q})}}$, that is induced by the permutation $P_{i_q}$, as follows
\begin{align}
        &w(e_{z_{(i_1, i_2 ,\ldots , i_{q-1})} z_{(i_1, \ldots , i_{q-1} , i_{q})}}) \notag \\
        &\qquad  \equiv \bra{z_{(i_1, \ldots , i_{q-1} , i_{q})}} \Pi_{k=1}^q D_{i_k} P_{i_k} \ket{z_{(i_1, i_2 ,\ldots , i_{q-1})}} \notag \\
        &\hspace{1.75in} = d^{(i_q)}_{z_{(i_1, \ldots , i_{q-1} , i_{q})}}\, .
\end{align}
Using this notation, we can emphasize that Hermiticity of the Hamiltonian implies that for every directed edge $e_{z_{(i_1, i_2 ,\ldots , i_{q-1})} z_{(i_1, \ldots , i_{q-1} , i_{q})}}$, there is an opposite edge $e_{z_{(i_1, \ldots , i_{q-1} , i_{q})} \,  z_{(i_1, i_2 ,\ldots , i_{q-1})}}$ with weight

\begin{align}
        &w(e_{z_{(i_1, \ldots , i_{q-1} , i_{q})} \,  z_{(i_1, i_2 ,\ldots , i_{q-1})}}) \notag \\
        &\equiv \bra{z_{(i_1, i_2 ,\ldots , i_{q-1})}} \Pi_{k=1}^q D_{i_k} P_{i_k} D_{i_k^{-1}} (P_{i_k})^{-1}\ket{z_{(i_1, \ldots , i_{q-1} , i_{q})}} \notag \\
        &\hspace{0.5in} = d^{(i_q^{-1})}_{z_{(i_1, \ldots , i_{q-1})}} = w(e_{z_{(i_1, i_2 ,\ldots , i_{q-1})} z_{(i_1, \ldots , i_{q-1} , i_{q})}})^* \notag \\
        &\hspace{1.35in}= \left(d^{(i_q)}_{z_{(i_1, \ldots , i_{q-1} , i_{q})}}\right)^*\, .
\end{align}

From here on, we will simply refer to the weights of the edges $d^{(i)}_z$ and will avoid using formal graph notation $w(e_{z_{\ldots}})$ to refer to the weight of an edge on the computational state graph.

\subsubsection{Closed walks on \texorpdfstring{$G_H$}{G H}}

\begin{definition}[Closed walks on $G_H$]
A \textit{closed walk of length} \( q \) on \( G_H \) is defined as a sequence of diagonals and permutation operators applied to a computational basis state \( \ket{z} \) of the form
\[
\mathcal{S}_{\mathbf{i}_q} = D_{i_q} P_{i_q} \cdots D_{i_2} P_{i_2} D_{i_1} P_{i_1},
\]
such that the final state obtained by applying \( \mathcal{S}_{\mathbf{i}_q} \) to \( \ket{z} \) is proportional to \( \ket{z} \). We denote the associated permutation string as 
\[
S_{\mathbf{i}_q} := P_{i_q} \cdots P_{i_2} P_{i_1},
\]
and we say that \( \mathcal{S}_{\mathbf{i}_q} \) \emph{generates a closed walk} on \( G_H \) starting and ending at \( \ket{z} \) if the associated permutation string is equivalent to identity
\[
S_{\mathbf{i}_q} = \mathds{1}.
\]
In other words, the action of \( \mathcal{S}_{\mathbf{i}_q} \) maps \( \ket{z} \) back to itself (up to a scalar factor). The corresponding \textbf{weight} of the closed walk is defined by
\begin{align}
        \mathcal{D}_{(z \, , \, \mathcal{S}_{\mathbf{i}_q})} &:= \bra{z} \mathcal{S}_{\mathbf{i}_q} \ket{z} \notag \\
        &= d^{(i_1)}_{z_{(i_1)}} \cdot d^{(i_2)}_{z_{(i_1, i_2)}} \cdots d^{(i_q)}_{z_{(i_1, i_2, \ldots, i_q)}},
        \label{eq:ref_weight}
\end{align}
where each \( d^{(i_j)}_{z_{(\cdot)}} \) is the diagonal entry of \( D_{i_j} \) evaluated on the intermediate state obtained by applying the preceding permutations \( P_{i_1}, \ldots, P_{i_{j-1}} \) to \( \ket{z} \), that is,
\[
\ket{z_{(i_1, \ldots, i_j)}} := P_{i_{j-1}} \cdots P_{i_1} \ket{z}.
\]
\end{definition}

Consider two identity-equivalent sequences $\mathcal{S}_{\mathbf{i}_{q_1}}$ and $\mathcal{S}_{\mathbf{j}_{q_2}}$. Their multiplication induces a longer closed walk of length $q_1 + q_2$ on $G_H$. The sequence $\mathcal{S}_{{\bf k}_{q_1 + q_2}} := \mathcal{S}_{\mathbf{i}_{q_2}} \mathcal{S}_{\mathbf{j}_{q_1}} $ applied to the state $\ket{z}$ generates a closed walk with weight
\begin{align}
        \mathcal{D}_{(z \, , \, \mathcal{S}_{\mathbf{k}_q})} = \mathcal{D}_{(z \, , \, \mathcal{S}_{\mathbf{i}_{q_1}})} \mathcal{D}_{(z \, , \, \mathcal{S}_{\mathbf{k}_{q_2}})} \, .
\end{align}

\subsection{Fundamental cycles and cycle generators}

We start by providing a rigorous definition of what is meant by a fundamental cycle. 
\begin{definition} [Fundamental cycles]
        A \textit{fundamental cycle} of length $q$, namely $\mathcal{C}^{(z)}_{\mathbf{i}_q}$, with starting state $\ket{z}$, is a closed walk of weight $\mathcal{D}_{(z \, , \, \mathcal{S}_{\mathbf{i}_q})}$, generated by the sequence $\mathcal{S}_{{\bf i}_q}$, such that no two non-adjacent states $\ket{z_1}$ and $\ket{z_2}$ along the walk are actually adjacent on the graph $G_H$. In other words, a fundamental cycle $\mathcal{C}^{(z)}_{\mathbf{i}_q}$ does not have any non-adjacent states $\ket{z_{i_j}}$ and $\ket{z_{i_k}}$ for which there exists a $D_k P_k$ such that $\bra{z_{i_k}} D_k P_k \ket{z_{i_j}} \neq 0$.
\end{definition}

In discussing the fundamental cycles of a graph $G_H$, it is important to highlight the role of strings of identity-equivalent permutations that \textit{generate} the walk on the graph. Given a set of $M$ permutation matrices that span the off-diagonals of a given Hamiltonian $H$, we can find the set of smallest identity-equivalent strings, i.e., the \textit{fundamental generators}, efficiently via linear-algebra methods \cite{Babakhani_2025}. However, it is important to define exactly what these strings are with the following definition.

\begin{definition}[Fundamental generators]
        Given a PMR decomposition of a Hamiltonian $H$, the \textit{fundamental generators} of closed walks on the computational state graph $G_H$ of $H$ are the set of shortest strings of identity-equivalent permutations $S_{\boldsymbol{i}_q} := P_{i_1} \, P_{i_2} \, \ldots P_{i_q}\equiv \mathds{1}$, such that they cannot be obtained by multiplying two or more shorter identity-equivalent strings. As an example, consider the two strings $S_{i_3} := P_1 P_2 P_3 = \mathds{1}$ and $S_{j_3} = P_3 P_4 P_5 = \mathds{1}$. One can multiply the two strings and obtain a length-four identity-equivalent string $S_{4} = S_{i_3} S_{j_3} = P_1 P_2 P_4 P_5$, however, $S_4$ would not be considered a fundamental generator since it can be obtained from two shorter strings.
\end{definition}

It has been shown, in prior work \cite{Barash_2024,Babakhani_2025}, that QMC for spin-$1/2$ and higher-spin ($S>1/2$) systems can be entirely generated via the fundamental cycle generators of a Hamiltonian. In the following sections, we will show that we can define cost functions using the weights of the fundamental cycles to quantify the severity of the QMC sign problem in spin systems. This work is focused on addressing the sign problem for known spin-$1/2$ systems, but the methods outlined can be easily extended to higher-spin systems.

\subsubsection{Diameter of \texorpdfstring{$G_H$}{G H} for local Hamiltonians}

\begin{definition}[Graph Diameter]
Let $d(\ket{z}, \ket{z'})$ denote the length of the shortest unweighted path between the states $\ket{z}$ and $\ket{z'}$ in the set of nodes $V$ of the graph $G_H$. The \emph{diameter} of $G_H$, denoted $\mathrm{diam}(G_H)$, is defined as
\[
        \mathrm{diam}(G_H) = \max_{\ket{z}, \ket{z'} \in V} d(\ket{z}, \ket{z'}).
\]
That is, the diameter is the greatest unweighted distance between any two states in $G_H$.
\end{definition}

\begin{lemma}[Scaling of $\mathrm{diam}(G_H)$ for physical Hamiltonians]
\label{lemma:G_diam}
Let $H$ be a geometrically local $k$-body Hamiltonian acting on $N$ particles (or spins), and let $G_H$ be the associated computational state graph whose vertices are the computational basis states $\ket{z}$, with edges generated by the action of local permutation matrices corresponding to the off-diagonal terms of $H$. Then the diameter of $G_H$ satisfies
\[
        \mathrm{diam}(G_H) = O(N).
\]
In other words, the greatest distance between any two computational basis states on $G_H$ grows at most linearly with the system size $N$.
\end{lemma}

\begin{proof}
Consider two computational basis states $\ket{z}$ and $\ket{z'}$ differing in their spin values on $k$ consecutive (neighboring) sites and identical elsewhere. Let us also denote the set of all states that differ on the same $k$ spin values, but agree on all others, as the set $\mathcal{Z}_k$. To transform $\ket{z}$ into $\ket{z'}$ (i.e., to walk from $\ket{z}$ to $\ket{z'}$ on $G_H$), it suffices to apply a sequence of local permutation operators $D_k P_k$ supported on $\mathcal{Z}_k$.

If $H$ represents a physical Hamiltonian, i.e., its support radius extends to $k = O(1)$ neighboring spins, there exist $O(1)$ local operators capable of transforming any configuration of $k$ neighboring spins, $\mathcal{Z}_k$, into any other. For each such $\mathcal{Z}_k$, all $2^k$ possible local configurations (with the remaining $N-k$ spins fixed) either form a connected component of $G_H$ or belong to disjoint components. Thus, for any pair $\ket{z},\ket{z'}$ that belong to the same graph $G_H$, differing only on $k$ neighboring spins, the shortest path between them is $O(1)$.

To reach any general $\ket{z'}$ from $\ket{z}$, iterate this process along the spin chain (or lattice), modifying $k$ sites at a time. This requires $O(N)$ steps, each of which involves $O(1)$ applications of local operators, yielding an overall path of length $O(N)$.

If two states $\ket{z}$ and $\ket{z'}$ cannot be connected in this way (e.g., if certain quantum numbers or symmetries forbid transitions), they lie in distinct disconnected components of $G_H$ and their distance is infinite. Hence, the diameter of each connected component of $G_H$ is $O(N)$. This concludes the proof.
\end{proof}

\begin{theorem}[Largest fundamental cycles on $G_H$]
\label{thm:Q=O(N)}
    The length of the largest fundamental cycle of $G_H$, denoted by $Q$, is $O(N)$.
\end{theorem}
\begin{proof}
    This is a direct consequence of Lemma~\ref{lemma:G_diam}. The farthest distance between any two states $\ket{z}$ and $\ket{z'}$ on a fundamental cycle of $G_H$ is at most $d(\ket{z} , \ket{z'}) = O(N)$. This implies that there are at most $O(N)$ strings of $D_i P_i$ required to generate the cycle. Thus, the length of the largest fundamental cycle is at most $O(N)$.
\end{proof}

\section{Review of PMRQMC}
\label{sec:pmrqmc_rev}
Casting the Hamiltonian in PMR form, one can show that the partition function $Z=\tr\left[ e^{-\beta H} \right]$ can be written as ~\cite{Gupta_2020}
\begin{align}
    Z  = \tr\left[ e^{-\beta H} \right] = \sum_{z} \sum_{q=0}^{\infty} \sum_{\{S_{{\bf{i}}_q}\}} \langle z | S_{{\bf{i}}_q} | z \rangle W^{(z)}_{S_{\boldsymbol{i}_q}} \,.
    \label{eq:SSE3}
\end{align}
where $ W^{(z)}_{S_{\boldsymbol{i}_q}} = \mathcal{D}_{(z \, , \, \mathcal{S}_{\mathbf{i}_q})} e^{-\beta[E_{z_0},\ldots,E_{z_q}]}$.
The sum above is a double sum: over the set of all basis states $z$ and over all products of $q$ permutation operators  $S_{{\bf{i}}_q}=P_{i_q} \ldots P_{i_2}P_{i_1}$ with $q$ running from zero to infinity. The multi-index ${\bf{i}}_q=(i_1,\ldots,i_q)$ ranges over all combinations of products of permutation operators with $i_j=1 \ldots M$. The first term, $\mathcal{D}_{(z \, , \, \mathcal{S}_{\mathbf{i}_q})}$, in $W^{(z)}_{S_{\boldsymbol{i}_q}}$ is simply the weight of a closed walk generated by $\mathcal{S}_{\mathbf{i}_q}$, as introduce in the previous section. This is, as we described in the previous section, the weight of a closed walk on the computational basis state $G_H$.
The second term in each summand, $e^{-\beta[ E_{z_0}, \ldots, E_{z_q}]}$,  is called the divided differences of the function $F(\cdot) = e^{-\beta (\cdot)}$ with respect to the inputs $[ E_{z_0}, \ldots, E_{z_q}]$. The divided differences~\cite{Deboor_2005} of a function $F[\cdot]$ is defined as,
\begin{equation}
F[ E_{z_0}, \ldots, E_{z_q} ] \equiv \sum_{j=0}^{q} \frac{ F(E_{z_j}) }{ \prod_{k \neq j} ( E_{z_j} - E_{z_k} ) }.
\end{equation}
Here, $E_{z_i} = \langle z_i \vert D_0 \vert z_i \rangle$ are the diagonal entries of the Hamiltonian and $\{ S_{i_q} \}$ denotes the set of
all products of length $q$ of off-diagonal operators $P_j$. 
For a full derivation of the expansion in Eq.~\eqref{eq:SSE3}, refer to \cite{Gupta_2020}.

We note that as written, the weights  $W^{(z)}_{S_{\boldsymbol{i}_q}}$ are complex-valued, despite the partition function being real (and positive). Since for every configuration $\mathcal{C}=\{ |z\rangle, S_{{\bf{i}}_q}\}$ there is a conjugate configuration $\bar{\mathcal{C}}=\{ |z\rangle, S^{\dagger}_{{\bf i}_q}\}$\footnote{For $S_{{\bf{i}}_q} =P_{i_q} \ldots  P_{i_2} P_{i_1}$, the conjugate sequence is simply $S^\dagger_{{\bf i}_q} = P_{i_1}^{-1} P_{i_2}^{-1} \ldots P_{i_q}^{-1}$.} that produces the conjugate weight $W_{\bar{\mathcal{C}}}=\overline{W_{\mathcal{C}}} = (W_{\mathcal{C}})^*$, the imaginary contributions cancel out.  Expressed differently, the imaginary portions of complex-valued weights do not contribute to the partition function and may be disregarded altogether. 

Before we move on, we note that $\langle z| S_{{\bf{i}}_q} |z\rangle$ evaluates either to 1 or to zero.
Moreover, since the permutation matrices with the exception of $P_0$ have no fixed points, the condition $\langle z| S_{{\bf{i}}_q} |z\rangle=1$ implies $S_{{\bf{i}}_q}=\mathds{1}$, i.e., $S_{{\bf{i}}_q}$ must evaluate to the identity element $P_0$ (note that the identity element does not appear in the sequences $S_{{\bf{i}}_q}$).  The expansion in Eq.~\eqref{eq:SSE3} can thus be more succinctly rewritten as
\begin{align}
Z =\sum_{z}\sum_{S_{{\bf i}_q}=\mathds{1}}  \mathcal{D}_{(z \, , \, \mathcal{S}_{\mathbf{i}_q})}   e^{-\beta [E_{z_0},\ldots,E_{z_q}]} \,,
\label{eq:z}
\end{align}
i.e., as a sum over all computational basis states and permutation matrix products that evaluate to the identity matrix and the sequence of basis states $\{|z_i\rangle \}$ may be viewed as a \emph{closed walk} on the computational basis graph of a Hamiltonian, as highlighted in the previous section. 

We highlight the importance of \emph{fundamental cycles} in our analysis of the sign problem. As described formally, in the previous section, these are the smallest length cycle that cannot be made out of joining two smaller cycles. Below, is a figure demonstrating how a fundamental cycle of length $5$ looks like on the computational state graph.

\begin{figure}[H]
    \centering
    \includegraphics[width=7cm]{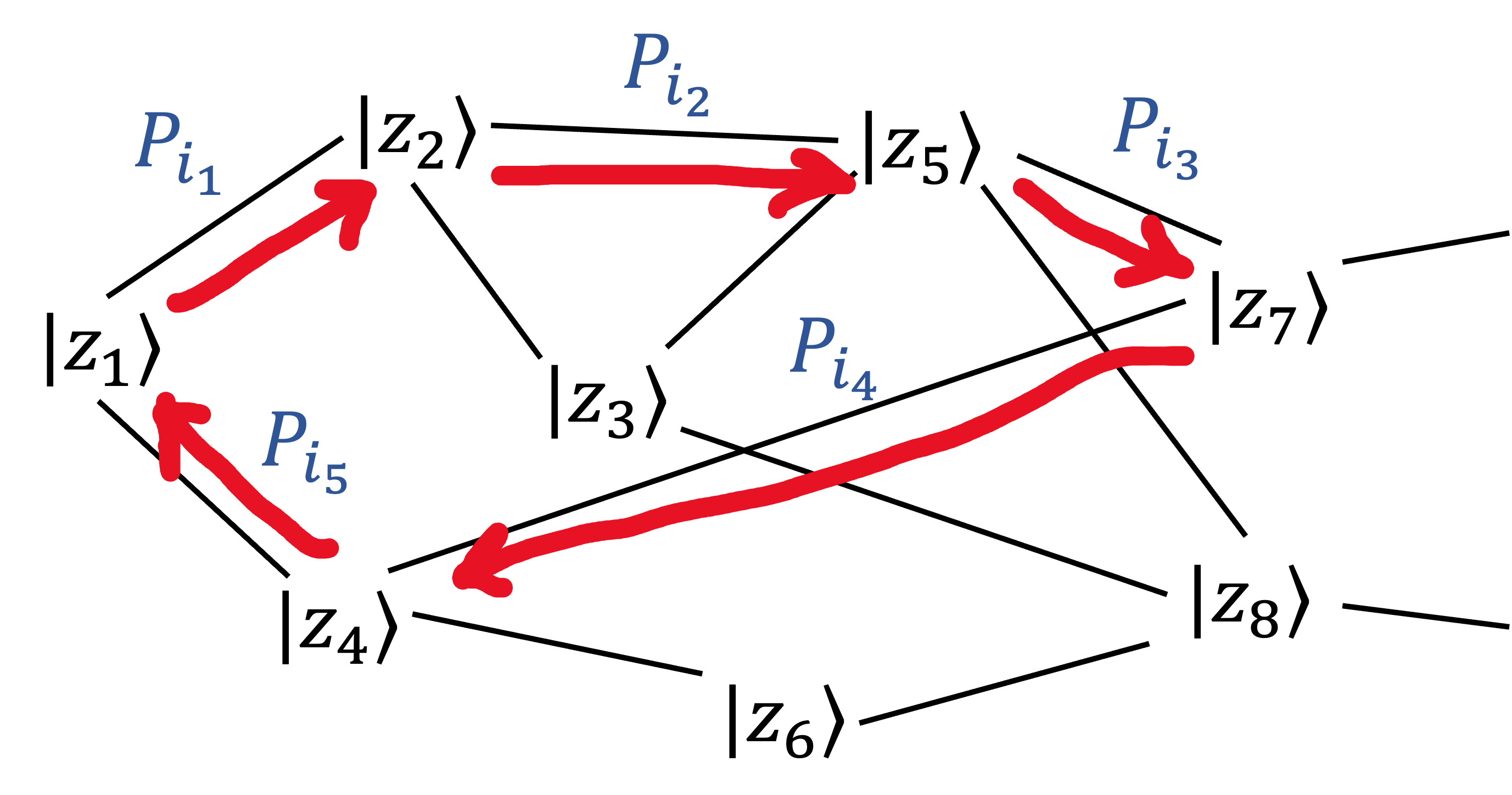}
    \caption{Fundamental cycle of length $5$, denoted by $\mathcal{C}_{i_5},$ produced by taking the string of $\Pi_{k=1}^5 P_{i_k}$ on state $\ket{z_1}$. The weight of the off-diagonals of this cycle is $W^{(z_1)}_{\mathcal{C}_{i_5}} = d^{(i_1)}_{(z_1)_{i_1}} d^{(i_2)}_{(z_2)_{i_2}} \ldots d^{(i_5)}_{(z_4)_{i_4}}$.}
    \label{fig:fund_cycle_5}
\end{figure}

Our goal in the next sections will be to address the sign problem by accounting for whether \emph{all} of the fundamental cycles on the computational state graph will satisfy the VGP condition given in Eq.\eqref{eq:VGP_angle}.

\section{Diagnostic measures for the sign problem}

\subsection{Average sign}

The usual measure that addresses the sign problem is the average sign of the QMC algorithm, which is defined as the following
\begin{align}
    \label{eq:avg_sign_def}
    \langle \text{sgn} \rangle  = \frac{\sum_{\mathcal{C}} W_{\mathcal{C}}}{\sum_{\mathcal{C}} |W_{\mathcal{C}}|} \, ,
\end{align}
where, as we highlighted at the end of previous section, $\mathcal{C}=\{ |z\rangle, S_{{\bf{i}}_q}\}$.

The average sign is commonly used to address the sign problem, with a sign free Hamiltonian taking on an average sign of $\langle \text{sgn} \rangle = 1$, and one with a severe sign problem, taking on a value close to $0$.

\subsection{Stoquasticity}

A Hamiltonian whose off-diagonal elements are all non-positive is said to be \textit{stoquastic}. Many attempts have been made to address the QMC sign via stoquasticity ~\cite{Hangleiter_2020,Klassen_2019,Klassen_2020}.
Consider a Hamiltonian $H = H_0 + H_{\text{off}}$, where $H_0$ and $H_{\text{off}}$ are its diagonal and off-diagonal parts respectively. Also, consider the following function
\begin{align}
    f_{\text{stoq.}}(H) := \| \,|H_{\text{off}}| - (- H_{\text{off}})\, \|_F \, ,
\end{align}
where $\|\cdot\|_F$ denotes the Frobenius norm. 

For a stoquastic Hamiltonian $H_{\text{stoq.}}$, we clearly have $f_{\text{stoq.}}(H_{\text{stoq.}}) = 0$. A generic Hamiltonian must seemingly satisfy $\mathds{D}(\mathds{D}-1)/2$, $\Theta(\mathds{D}^2)$ where $\mathds{D} = 2^N$, constraints in order to be stoquastic.
For local Hamiltonians that can be expressed in $O(poly(N))$ Pauli-strings, however, there is an efficient method to determine stoquasticity \cite{Klassen_2019}. This is partly why stoquasticity remains an attractive candidate to diagnose whether a Hamiltonian would suffer from QMC sign problem.

Additionally, there have been many attempts to curing and mitigating sign problematic Hamiltonians via local Clifford rotations \cite{Ioannou_2020}. In an interesting work by Marvian et al.~\cite{Marvian_2019}, it was shown that determining whether a Hamiltonian can be made stoquastic via local unitary transformations is QMA-hard in general. 

However, recent insights by Hen~\cite{Hen_2021} reveal that the class of sign-problem-free Hamiltonians is larger than the stoquastic class. They identify a broader category known as VGP Hamiltonians, which can yield non-negative QMC weights despite being non-stoquastic, due to destructive interference from the phases of the  edges of $G_H$ that contribute to the fundamental cycles.

We will next explore the VGP phase and propose a cost function that addresses the negativity of the weights of the QMC samples, without imposing strict stoquasticity conditions.

\subsection{Vanishing Geometric Phase (VGP)}
Using the computational state graph, $G_H$, of a given Hamiltonian $H$, the partition function can be expressed as the sum of weights of closed walks on the graph $G_H$.

Inspired by \cite{Hen_2021}, we consider addressing the sign problem via the weights of the fundamental cycles on the graph $G_H$. Since the divided difference term in Eq.~\eqref{eq:SSE3} introduces a scalar whose sign is only a function of the \textit{number of input variables} ($(-1)^q$ for $q$ edges and $q+1$ visited states), the weight of each fundamental cycle of length $q$, can be written 
\begin{align}
W^{(z)}_{\mathcal{C}_{i_q}} = (-1)^q \mathcal{D}_{(z , \mathcal{C}_{i_q})} = |\mathcal{D}_{(z , \mathcal{C}_{i_q})}| e^{i (\Phi^{(z)}_{\mathcal{C}_{i_q}} + q\pi)} \, .
\label{eq:fc_weight}
\end{align}

If \textit{all} fundamental cycles of a Hamiltonian satisfy the following condition
\begin{align}
    \Phi^{(z)}_{\mathcal{C}_{i_q}} + q\pi \equiv 0 \mod 2\pi \, ,
    \label{eq:VGP_angle}
\end{align}
it is said that the Hamiltonian is VGP, and is sign problem free. This is a more general condition than \textit{stoquasticity}. We again refer the reader to \cite{Hen_2021} for detailed discussion and derivation for this condition. 

The diagram below demonstrates the space of unitaries that will make an initial sign-problematic Hamiltonian into a sign-free one.

\begin{figure}[h]
    \centering
    \includegraphics[width=0.5\linewidth]{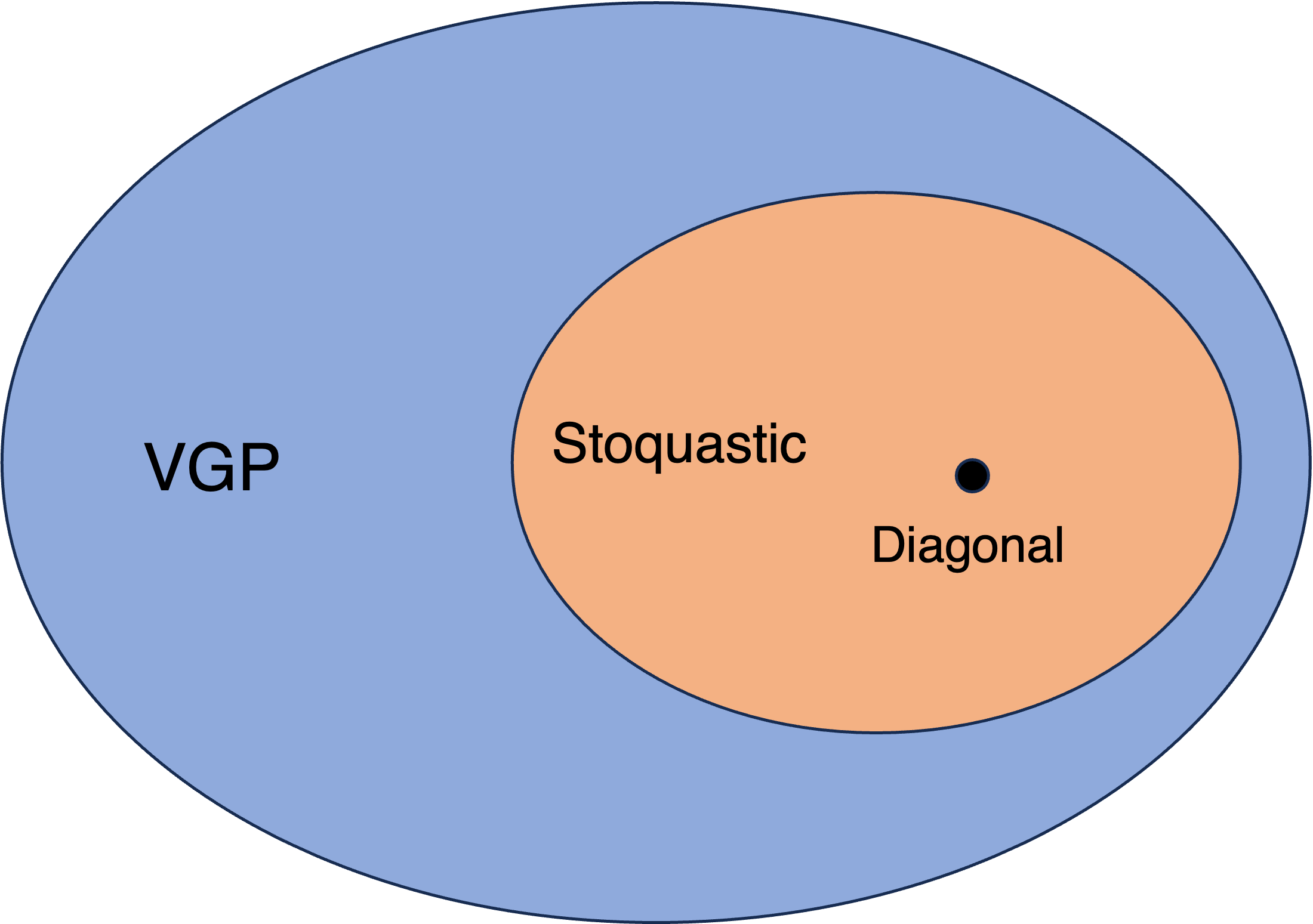}
    \caption{Diagram showing classes of Hamiltonian that are sign-problem free.}
    \label{fig:enter-label}
\end{figure}

As discussed in earlier sections, a VGP Hamiltonian can be mapped to a stoquastic Hamiltonian via conjugation by a suitable diagonal unitary. However, finding such a transformation using only local Pauli and Clifford group operations is generally infeasible, since the Clifford group does not form a universal gate set for the group of unitary operators. Consequently, there exist diagonal unitaries outside the Clifford group that cannot be constructed from local Clifford rotations alone.

In order to generate rotations that will turn a Hamiltonian into VGP, we will define a metric with the goal of enforcing Eq. (\ref{eq:VGP_angle}). 
\begin{definition}[VGP indicator]
    \label{def:VGP_ind}
    The following function
    \begin{align}
        f_{\text{VGP}}(H) = \sum_{q=3}^{Q} \sum_{\mathcal{C}_q \in G_H} |\mathcal{D}^{(z)}_{\mathcal{C}_{q}}| (1 - \cos(\Phi_{\mathcal{C}_{q}})) \, ,
        \label{eq:VGP_cost}
    \end{align}
    where $|\mathcal{D}^{(z)}_{\mathcal{C}_{q}}|$ is the magnitude of the off-diagonal weights of a fundamental cycle $\mathcal{C}_q$, given by Eq. (\ref{eq:fc_weight}). This indicator function measures how far away a Hamiltonian representation is from being VGP. We note that Eq. (\ref{eq:VGP_cost}) contains a sum over the index $q$ that goes from $3$ to $Q$, where $Q$ is the length of the largest fundamental cycle for the computational graph representation of the Hamiltonian. $C_q$ is the index denoting the sum over the fundamental cycles of length of $q$. We would like to highlight the fact that there are length $2$ fundamental cycles, i.e. a made out of a pair of $D_iP_i$ (in spin-$1/2$ case, $P_i^2 = \mathds{1}$). However, weights of length-$2$ fundamental cycles are always non-negative due to the fact that $\bra{z} D_i P_i D_i P_i \ket{z} = \big(D_i P_i \ket{z}\big)^{\dagger} \big(D_i P_i \ket{z}\big) \geq 0$.
\end{definition}

\begin{corollary}[$Q = O(N)$ for physical $H$]
     For a physical spin Hamiltonian $H$, describing the interaction of $N$ spins with local interaction, the largest fundamental cycle on the computational state graph of $H$, namely $G(H)$, is of order $O(N)$. In other words, $Q = O(N)$ in the definition of $f_{VGP}(H)$ Eq.\eqref{eq:VGP_cost}.
\end{corollary}
\begin{proof}
    This is exactly the result of Theorem \ref{thm:Q=O(N)}, restated for the $Q$ being the largest fundamental cycle of $G_H$ for a physical (geometrically local) Hamiltonian subject to the focus of this paper.
\end{proof}

The corollary above will be used later to argue that the hardness of determining VGP for a local Hamiltonian with no structure on the diagonal and permutation matrices in the PMR form of the Hamiltonian will be exponentially difficult. However, for Hamiltonian with certain structure on the permutation matrices and on the radius of support of the diagonal matrices, one is able to determine if a Hamiltonian is VGP efficiently. We will demonstrate these instances in subsection~\ref{subsec:VGP_complexity}.  
\subsection{Symmetries of \texorpdfstring{$f_{\text{VGP}}$}{f VGP}}

Here, we emphasize a few properties for $f_{\text{VGP}}(H)$. These properties are nothing but the symmetries of $f_{\text{VGP}}$, i.e. the Hamiltonian basis transformations that leave the sign problem unchanged. In doing so, we highlight a class of unitary transformations, namely diagonal unitaries and Pauli matrices, that will not alter or cure the sign problem.

\begin{lemma}[Diagonal Unitaries]
    \label{lem:diag_conj}
    For a diagonal unitary $\Phi$, i.e. $\Phi = \text{diag}(e^{i\phi_1} , e^{i\phi_2} , \ldots , e^{i\phi_{\mathds{D}}})$, $f_{\text{VGP}}(\Phi H \Phi^\dagger) = f_{\text{VGP}}(H)$.
\end{lemma}

\begin{proof}
    Restating the analysis done in \cite{Hen_2021}, the diagonal weights $D^{(z)}_{\mathcal{C}_{q}}$ are nothing but the product of off-diagonal elements of the Hamiltonian $H_{z_i z_j}$, written in matrix element form below
    \begin{align}
        D^{(z)}_{\mathcal{C}_{q}} = H_{z z_1}\, H_{z_1 z_2} \ldots H_{z_q z} \, .
        \label{eq:Diagonal_unit_weight}
    \end{align}
    Note that after a diagonal transformation, the new matrix elements are
    \begin{align}
        \tilde{H}_{z_i z_j} = e^{i (\phi_{z_i} - \phi_{z_j})} H_{z_i z_j} \, .
    \end{align}
    Substituting this back into Eq. \eqref{eq:Diagonal_unit_weight} above, we see that 
    \begin{align}
        \tilde{\mathcal{D}}^{(z)}_{\mathcal{C}_{q}}& = \tilde{H}_{z z_1}\, \tilde{H}_{z_1 z_2} \ldots \tilde{H}_{z_q z} \notag \\
        &= e^{i(\phi_{z_i} - \phi_{z_j} + \phi_{z_j}  \ldots -\phi_{z_i})} \mathcal{D}^{(z)}_{\mathcal{C}_{q}} =  \mathcal{D}^{(z)}_{\mathcal{C}_{q}} \, .
    \end{align}
    This shows that since the weights remain the same, from Eqs. \eqref{eq:fc_weight} and \eqref{eq:VGP_cost}, so does the value of $f_{\text{VGP}}(H)$.
\end{proof}

\begin{lemma}
    \label{lem:P_conj}
    For spin-$1/2$ systems, given a permutation matrix $P = \Pi_{i} X_i$, $f_{\text{VGP}}(P\,H\,P^\dagger) = f_{\text{VGP}}(P\,H\, P) = f_{\text{VGP}}(H)$.
\end{lemma}

\begin{proof}
    Writing the transformed Hamiltonian in PMR form $PHP = \tilde{D}_0 + \tilde{H}_{\text{off.}}$, where $\tilde{D}_0 = P D_0 P$ is a transformed diagonal matrix from the original diagonal $D_0$ of $H$, and $\tilde{H}_{\text{off.}} = P\,\sum_i D_i P_i\, P$. Any off-diagonal fundamental cycle is of the form
    \begin{align}
        \bra{z} (PD_{i_1} P_{i_1}P) \, (PD_{i_2} P_{i_2}P)\ldots (PD_{i_q} P_{i_q}P)\ket{z} \notag \\
        =  \bra{z'} D_{i_1} P_{i_1}\, D_{i_2} P_{i_2}\ldots D_{i_q} P_{i_q}\ket{z'} \, .
    \end{align}
    where $\ket{z'} := P \ket{z}$. Thus, $P$ conjugation simply maps $\mathcal{D}^{(z)}_{\mathcal{C}_{q}}$ to $\mathcal{D}^{(z')}_{\mathcal{C}_{q}}$. This is nothing but relabeling our computational basis states, and since we are summing over all fundamental walks, relabeling does not change the value of $f_{\text{VGP}}(H)$. This concludes our proof.
\end{proof}

\begin{theorem}[No sign curing with Pauli operators]
    If $f_{\text{VGP}}(H) \neq 0$, there cannot be any Pauli operator $\mathcal{P}$, such that $f_{\text{VGP}}(\mathcal{P} H \mathcal{P}^\dagger) = 0$.
\end{theorem}
\begin{proof}
    This theorem follows immediately from Lemmas \ref{lem:diag_conj} and \ref{lem:P_conj} by noting that any Pauli operator $\mathcal{P} = \Phi P$, where $\Phi$ is a $Z$-string up to a phase.
\end{proof}

\subsection{Complexity of determining VGP}
\label{subsec:VGP_complexity}

\begin{theorem}[Determining VGP is hard for general physical Hamiltonians]
\label{thm:fVGP_hard}
    Given a geometrically local Hamiltonian $H$, determining if $f_{\text{VGP}}(H) = 0$, requires exponential steps in $N$ in the worst case. This implies that determining if a Hamiltonian is VGP is an exponentially difficult task.
\end{theorem}

\begin{proof}
    Eq. \ref{eq:VGP_cost} implies that in order to determine whether a Hamiltonian is VGP, i.e. $f_{\text{VGP}}(H) = 0$, one needs to determine whether \textit{all} fundamental cycles of length $q$ on the computational state graph have a complex phase of $q\pi$ (modulo $2\pi$). 
    We note that evaluating the weights of fundamental cycles of length $Q$ requires considering all possible combinations of diagonal components $D_{i_k}$ for each $P_{i_k}$ that will form an identity equivalent string of length $Q$. Determining if $f_{\text{VGP}}(H)=0$ requires evaluating fundamental cycles of length up to $Q = O(N)$. If each $D_{i_k}$ can have more than a single value acting on the computational basis states, i.e. if there are more than a single possible value for $d^{(z)}_{i_k}$ for $D_{i_k} \ket{z} = d^{(z)}_{i_k}\ket{z}$, then there will be exponentially many possible different length $Q=O(N)$ weights that one must evaluate.
\end{proof}

Below, we discuss a specific class of sparse Hamiltonians for which determining VGPness is $O(cN)$, where $c$, an $N$-independent constant, is a function of the number of Pauli-$Z$ strings required to express the diagonal matrices for each PMR input.

\begin{theorem}[Efficient VGP classification for certain sparse Hamiltonians]
    \label{thm:DX_VGP}
    Consider a \emph{geometrically local} Hamiltonian of the form $H = \sum_i D_i X_i$, with $D_i$ being a general sum of Pauli-$Z$ strings, with $k=O(1)$ radius of support. One can efficiently determine if $H$ is VGP in $O(N poly(2^{k}) )$ time.
\end{theorem}
\begin{proof}
    For Hamiltonians like this, the only cycles come from repetition of $X_j$ even times. When $D_j X_j$ appears in a cycle next to itself like the following
    \begin{align}
        \bra{z} (\Pi_{i_k} &D_{i_k} X_{i_k}) D_j X_j D_j X_j (\Pi_{i_{k'}} D_{i_{k'}} X_{i_{k'}}) \ket{z} \notag \\
        &= (d_{\vec{k}}) \, |d^{(i)}_{z_{(i_1' , i_2' , \ldots, i_k')}}|^2 (d_{\vec{k'}}) \, .
    \end{align}
    Here, $d_{\vec{k'}}$ denotes the weight of the diagonal terms coming from the rest of the diagonal and permutations. The contribution of a pair of $D_j X_j$ to the weight of the cycle is simply a real positive number. Hence, a term like this does not contribute to the sign problem as it maintains the condition given by Eq.~\eqref{eq:VGP_angle}.
    
    Next, consider a pair of $D_j X_j$ that is separated by at least another diagonal and permutation, namely $D_l X_l$. In this case, one obtains 
    \begin{align}
        \bra{z} &(\Pi_{i_k} D_{i_k} X_{i_k}) D_j X_j D_l X_l D_j X_j (\Pi_{i_{k'}} D_{i_{k'}} X_{i_{k'}}) \ket{z} \notag \\
        &= (d_{\vec{k} , l}) \, d^{(i)}_{z_{(i_1' , i_2' , \ldots, i_k' , l)}} (d^{(l)}_{z_{(i_1' , i_2' , \ldots, i_k' , l)}})^* (d^{(i)}_{z_{(i_1' , i_2' , \ldots, i_k')}})^*  (d_{\vec{k'}}) \, .
    \end{align}
    Thus, instances where a pair of $D_i X_i$ is separated by at least another diagonal and permutation $D_l X_l$, such that the support of the operator $D_i$ includes the $l$-th spin, i.e. $[D_i , X_l] \neq 0$, are the only way in which sign problematic weights could arise.
    
    Additionally, a general expression for a diagonal matrix $D_j$ with finite radius of support $k=O(1)$, can be written as follows
    \begin{align}
    \label{eq:Dj_terms}
        D_j \;=\; \sum_{S \subseteq \mathrm{supp}(D_j)} 
        h^{(j)}_{S} \prod_{l \in S} Z_l \,.
    \end{align}
    We can have linear combination of up to $k$-local strings of $Z$ since the radius of support of the diagonal operator is up to $k$ neighboring spins. $h^{(j)}_{S}$ is the coefficient for the linear combination for each $|S|$-local term, and there are at most $2^k$ such coefficients.
    
    This diagonal matrix, has at most $O(2^{k})$ eigenvalues (due to the finite $k$-spin support of the diagonal matrix). Hence, there are $O(2^k)$ different diagonal terms $d^{(j)}_{\vec{\alpha}}$, and one can write out the eigenvalues of $D_j$ as
    \begin{align}
    \label{eq:Dj_eigens}
    d^{(j)}_{\vec{\alpha}}
    = \sum_{S \subseteq \mathrm{supp}(D_j)} 
    h^{(j)}_{S} \prod_{\ell \in S} \alpha_\ell \,.
    \end{align}

    where $\alpha_{j_{n}} \in \{-1 , 1\}$. 
    We reiterate again that the weight for a single pair on a state $\ket{z}$ can be written as follows
    \begin{align}
        \bra{z} D_i X_i D_i X_i \ket{z} = (d^{(i)}_{\vec{\alpha}})^* d^{(i)}_{\vec{\alpha}} \geq 0 \,.
    \end{align}
    Furthermore, adopting the notation of Eq.~\eqref{eq:Dj_eigens}, the weight for two pairs $D_i X_i$ and $D_j X_j$ can be written as
    \begin{align}
        \bra{z} D_i X_i D_j X_j D_i X_i D_j X_j \ket{z} = (d^{(i)}_{\vec{\alpha}_i})^* (d^{(j)}_{\vec{\alpha}_j})^*d^{(i)}_{\vec{\alpha}'_i} d^{(j)}_{\vec{\alpha}'_j}\,.
    \end{align}
    However, if $[D_j , X_i] = [D_i , X_j] = 0$, $\vec{\alpha}'_j = \vec{\alpha}_j$ and $\vec{\alpha}'_i = \vec{\alpha}_i$, so the considered cycle is free of the sign problem, where $[\, ,\, ]$ denotes commutation operation between two operators.

    Since, every $D_i$ can have support over at most $k$ other local spins, in order to determine if a Hamiltonian is VGP, we need to verify $O(k)$ statements for each $D_i X_i$, up to cycles of length $2q$ and determine if 
    \begin{align}
    \label{eq:D_argcheck}
    \arg\Bigg[\Big(\Pi_{j=1}^{q-1} (d^{(j)}_{\vec{\alpha}_j})^* d^{(j)}_{\vec{\alpha}'_j}) \Big)(d^{(i)}_{\vec{\alpha}'_i})^*d^{(i)}_{\vec{\alpha}_i}  \Bigg] \equiv 0 \mod 2\pi \, .
    \end{align}
    Note that since the length of all fundamental cycles is even, we have dropped the factor $2q\pi$ from Eq.~\eqref{eq:VGP_angle} from the mod statement above, as it is a multiple of $2\pi$.
    
    Of course, the largest $q$ is determined by the number of $D_j$s that have common support with $X_i$, which is at most $k$. 
    Different orderings of the operators also produce different values of $\vec{\alpha}'_j$ for each $D_j$, but at most we have $O(poly(2^{k}))$ statements to verify for each $D_i X_i$ \footnote{This is because $D$ has support over at most $k$ local spins, so there are $k$ components to a vector of coefficients $\vec{\alpha}$, as given by Eq.~\eqref{eq:Dj_eigens}. By considering combination of all possible $\vec{\alpha}_j$ and their flipped counterparts $\vec{\alpha}'_j$, there are at most $2^{k}(2^k-1)$ pairings for each $d$ term. Depending on the order of the graph, and the type of connectivity of the lattice, there should be $O(poly(2^{k}))$ statements with different $\vec{\alpha}$ and their flipped counterparts $\vec{\alpha}'$.}.

    Since there are $O(N)$ terms in the Hamiltonian, we have to perform $O(N poly(2^{k}))$ statement verifications. This concludes our proof for the efficiency of verification of VGPness for Hamiltonians of this class.
\end{proof}

\begin{corollary}
    Consider Hamiltonian of the form $H = \sum_i D_i P_i$, with $D_i$ being a general sum of Pauli-$Z$ strings, with an $k=O(1)$ radius of support. If there are no sets of existing $\{P_j\}$ such that $\Pi_j P_j = \mathds{1}$, $H$ can be determined to be VGP in $O(N poly(2^{k}))$ steps.
\end{corollary}
\begin{proof}
    The proof follows similarly to the proof of theorem \ref{thm:DX_VGP}. In this instance, since $P_i$ might be of different form than a single Pauli-$X$ on the $i$-th spin, checking whether Eq. (\ref{eq:D_argcheck}) is satisfied for each diagonal might require less steps, since not every single spin flip need be considered. So, in general the number of verifications is still $O(poly(2^{k}))$. Thus, VGP-ness of $H$ of this type, can also be determined in $O(N poly(2^{k}))$.
\end{proof}

\subsection{Other measures}
\label{sec:measures}
In this section, we will define two functions of the Hamiltonian $f:\operatorname{Herm}(\mathds{D}) \mapsto \mathds{R}^+$, that `measures' the severity of the sign problem. Even though these metrics might seem like another way to express the average sign, they provide a convenient way to make scaling arguments and address the severity of the sign problem under change of basis, as we will discuss in Section~\ref{sec:avg_sign_stats}.
\subsubsection{Off-diagonal absolutes (bosonization)}
Instead of using $f_{\text{VGP}}(H)$, we can use a similar function that sums up all the weights of the fundamental and non-fundamental cycles using the following function
\begin{align}
    \label{eq:f_eta}
    f_\eta(H) = \tr(e^{\eta |H_{\text{off}}|}) - \tr(e^{-\eta H_{\text{off}}}) \, .
\end{align}
The parameter $\eta \geq 0$ is a real number, that plays the role of the inverse temperature. The specific value of $\eta$ is not of importance: we have only kept it in order to show how $f_\eta$ scales as a function of $\eta$ (when $H$ is not VGP), in order to make inferences on how the average sign decays as a function of an inverse temperature like parameter.
The motivation for defining this as a metric in our setting was to have a numerically convenient metric to determine whether a Hamiltonian is VGP or not, in our numerics, instead of $f_{\text{VGP}}$. Even though $f_{\text{VGP}}$ is always less computationally expensive, $f_\eta$ is readily available in every programming environment. 

Now we will explain how $f_\eta$ and $f_{\text{VGP}}$ are related. Upon expanding the exponential and using the PMR formalism to express the traces as sums over weights of closed walks, we note that Eq.~\eqref{eq:f_eta} results in a similar expression as Eq.~\eqref{eq:VGP_cost}, but with division of $q!$ for closed walks of length $q$. In addition, ~\eqref{eq:f_eta} accounts for all cycles of any length (i.e. $Q \rightarrow \infty$), whereas in order to determine whether a Hamiltonian is truly VGP, one needs to only consider $Q=O(N)$. 

Since $f_\eta(H) = 0$ implies that \emph{all} closed cycles (fundamental and non-fundamental) are VGP, it should be noted that $f_\eta(H) = 0$, if and only if $f_{\text{VGP}}(H) = 0$. In other words, if the fundamental cycles of graph representation of $H$ are all VGP, \emph{all} cycles of the graph are VGP, and vice versa.

It is important to note that stoquasticity requires $\Theta(\mathds{D}^2)$ constraints on all the off-diagonal elements of the Hamiltonian. However, $f_\eta(H)$ does not impose such strict restrictions as one can find instances of non-stoquastic Hamiltonians that satisfy $f_\eta(H) = 0$, as we will see in Section~\ref{sec:VGPorStoq}.

\subsubsection{A measure of off-diagonal interference}

Define $\mathcal{M} : \operatorname{Herm}(\mathds{D}) \mapsto \operatorname{Herm}(\mathds{D})$ as a self-map on the set of $\mathds{D} \times \mathds{D}$ Hermitian matrices by $\mathcal{M}(X) := |X_{\text{off}}| - \text{diag}(X)$. Here, $X_\text{off} := X - \text{diag}(X)$ and $| \cdot |$ denotes entry-wise absolute value. That is, $|X_\text{off}|$ is the entry-wise absolute value of the off-diagonal of a given Hermitian matrix $X$. Next, define $\mathcal{F}_X(U) : \operatorname{Herm}(\mathds{D}) \mapsto \mathbb{R}$ by $\mathcal{F}_X(U) := \tr[e^{\mathcal{M}(UXU^\dagger)}]$. It was brought to our attention that a previous attempt was made to consider the spectrum of $\mathcal{M}(H)$ in~\cite{Crosson_2020}. Even though the attempt and motivation in the paper by Crosson et al. are quite similar to ours, we are taking a slightly different angle in that we are studying the expectation value of a VGP inducing transformation via random unitaries and the concentration of sign curing/mitigating unitaries in the random unitary space.

Consider a physical Hamiltonian $H$ in a given basis as before.
Notice that if we cast into PMR form $\mathcal{M}(H) = -D_0 + \sum_{j=1}^{M'} |D_j|P_j$, we obtain an expansion with the same permutation operators as in the original PMR decomposition of $H$, given by Eq.~\eqref{eq:PMR}. Hence the computational state graphs $\mathcal{M}(H)$ and $H$ share the same closed walks $\{S_{\textbf{i}_q}\}$. If we expand $\mathcal{F}_H(\mathds{1})$ as in equation~\eqref{eq:z}, we can see that for any closed walk $S_{\textbf{i}_q}$ with weight $\mathcal{D}_{(z, S_{\textbf{i}_q})}$ in~\eqref{eq:z}, the closed walk has weight $|\mathcal{D}_{(z, S_{\textbf{i}_q})}|$, analogous to the expansion of $f_\eta(H)$ (a consequence of taking the entry-wise absolute value of the off-diagonal). Also see that the divided difference terms in the summand of Eq.~\eqref{eq:SSE3} remain the same.

We also highlight that $\frac{Z(H)}{\mathcal{F}_H(\mathds{1})}$ is precisely the average sign of the QMC defined in~\eqref{eq:avg_sign_def}, when the inverse temperature $\beta$ is $1$. This definition of $\mathcal{F}_H$---while perhaps not computationally tractable for large dimensions--- allows us to gain insight into the mean behavior of the average sign in section ~\ref{sec:avg_sign_stats} across different basis transformations. This insight allows us to formalize the notion that a random search over all bases to optimize the average sign is hopeless.

\section{Scaling of \texorpdfstring{$f_{\text{VGP}}$}{f VGP} and sign-problem severity}
\label{sec:VGPscaling_signseverity}
We would like to highlight how to quantify the severity of the sign problem through scaling analysis. For this section, we will focus on the behavior of $f_{\eta}(H)$.

So far, we have established that if any of the measures specified in section \ref{sec:measures} are strictly zero, then the system is sign-problem free. But, what if $f_{\eta}(H)$ is not zero? Is there a regime in which even though $f_{\eta}(H)$ is not zero, the sign-problem is not severe and will \textit{not} plague us when scaling the size of a system or inverse temperature $\beta$?

If we expand $f_{\eta}(H)$ using our PMR formalism, and adopt the same notation as in Eq. (\ref{eq:VGP_cost}), we get
\begin{align}
    f_{\eta}(H) = \sum_z \sum_{q=3}^\infty \sum_{\mathcal{S}_{{\bf i}_q}} \frac{\eta^q}{q!} |\mathcal{D}_{(z , \mathcal{S}_{{\bf i}_q})}|\left(1 - \cos(\Phi^{(z)}_{\mathcal{S}_{{\bf i}_q}})\right) \, .
    \label{eq:fVGP_expansion}
\end{align}

As highlighted in Eq.~\eqref{eq:VGP_angle}, a Hamiltonian is sign-problem free if and only if $\Phi^{(z)}_{\mathcal{S}_{{\bf i}_q}} + q\pi \equiv 0 \mod 2\pi$ for all closed paths $\mathcal{S}_{{\bf i}_q}$. To discuss how $f_{\eta}(H)$ scales with $N$ and $\eta$, we consider the term without the $\left(1 - \cos(\Phi^{(z)}_{\mathcal{C}_{{\bf i}_q}})\right)$ factor, i.e. $\tr(e^{\eta |H_{\text{off}}|}) = \sum_z \sum_{q=3}^\infty \sum_{\mathcal{S}_{{\bf i}_q}} \frac{\eta^q}{q!} \left| \mathcal{D}_{(z , \mathcal{S}_{{\bf i}_q})} \right|$. 
The scaling of this term depends on the number of non-zero off-diagonal matrix elements. In most physically local Hamiltonians, where the number interactions scale as $O(N)$\footnote{Even for long-range interactions scaling as $O(N^2)$, the dominant contributions often still come from a sparse subset of local terms, leading to an effective $O(N)$ scaling. However, there are glassy systems in which the number of interactions are $O(N^2)$ or in general $O(poly(N))$, in which case $\tr(e^{\eta |H_{\text{off}}|})$ scales as $O(\exp(\eta poly(N)))$.}

In order to fully simulate a Hamiltonian, we must be able to simulate everything up to ground state properties, which requires inverse temperatures $\beta$ that are much smaller than the energy gap between the ground and the first excited state. For most physically interesting Hamiltonians, this energy gap depends on $N$ \cite{Albash_2018, Schaller_2008}. However, let us focus on the scaling of $f_{\eta}(H)$ as the two main parameters: system size $N$ and inverse temperature $\beta$. 

\subsection{System size}
\label{subsec:scaling_N}
An important motivation to understand the scaling behavior of the average sign of QMC is to ascertain the ability to study thermodynamic behavior. Here, we have constructed a toy model to illustrate the scaling behavior of average sign from the VGP point of view.
We describe how the scaling of average sign as a function of system size can be determined by the ratio of the VGP vs non-VGP fundamental cycles. 

Consider the equal weight Heisenberg model on a ladder triangular lattice with coefficient $-1$ written as follows
\begin{align}
    H_{\text{Heis.}} &= -\sum_{\langle i ,j \rangle \in G^{\rightarrow}_{\Delta}} (\vec{S}_i \cdot \vec{S}_j) 
\end{align}
where we have used the notation $G^{\rightarrow}_{\Delta}$ to indicate a triangular connection between spins, extended on a ladder (quasi-1$D$). This Hamiltonian has no sign problem. 

Now, let us turn the coefficient of the diagonal connections from $-1$ to $+1$. Now, we are certainly adding fundamental cycles that are non-VGP, i.e. the complex phase produced by their diagonal weights violates the VGP condition given by Eq.\eqref{eq:VGP_angle}. These edges can be thought of as \emph{VGP defects}, where the addition of a single such edge would produce numerous non-VGP fundamental cycles. We will use such defect edges to construct a non-VGP Hamiltonian where we have control over the number of non-VGP cycles produced by positioning defect edges on the ladder lattice. Thus we will consider analyzing the scaling of the average sign for $H^{(\text{d})}_{\text{Heis.}}$ defined as follows  
\begin{align}
    \label{eq:Hamheis_defect}
    H^{(\text{d})}_{\text{Heis.}} &= H_{\text{Heis.}}+ H_d \\ 
    H_{\text{d}} &= 2 \sum_{\langle i , j \rangle \in \text{d}} (\vec{S}_i \cdot \vec{S}_j) \, .
\end{align}
Here $\text{d}$ defines a set of diagonal `defect' edges that create sign-problematic weights.

\begin{figure}[H]
    \centering
    \includegraphics[width=7cm]{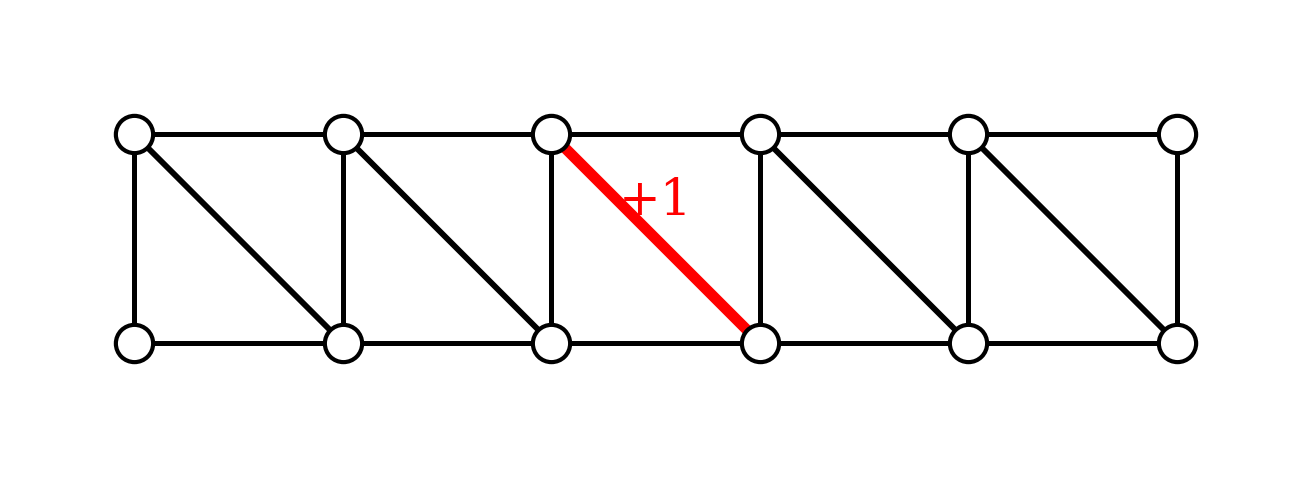}
    \caption{Figure of the ladder triangular lattice with a \emph{defect}. The defect edge is highlighted in the figure.}
    \label{fig:placeholder}
\end{figure}

We will consider two cases: a single defect edge $N_d = 1$ in the middle of the ladder, and a case where two defect edges $N_d = 2$ spread out evenly from the two ends of the ladder. What we will show is that adding a few constant number of defects along the chain does not change the average sign as a function of $N$. We also illustrate how to estimate the average sign by counting the number of VGP cycles. This of course works only if different cycles have almost equal magnitude of weights.

Placing a defect in the middle of the lattice, will produce many non-VGP fundamental cycles. Let us consider the case where $\beta$ is extremely low so that the mean cycle length ($\langle q \rangle_{\text{QMC}}$) sampled in QMC calculations is less than $3$. In this instance, the partition function value is determined by the value of the weights of all length-$2$ cycles. And since length-$2$ cycles are always VGP, the average sign should be almost $1$. We illustrate this with the following plot of the average sign ($\langle \text{sgn}(W) \rangle$) vs $N$, for $\beta = 0.1$. 
\begin{figure}[H]
    \centering
    \includegraphics[width=6cm]{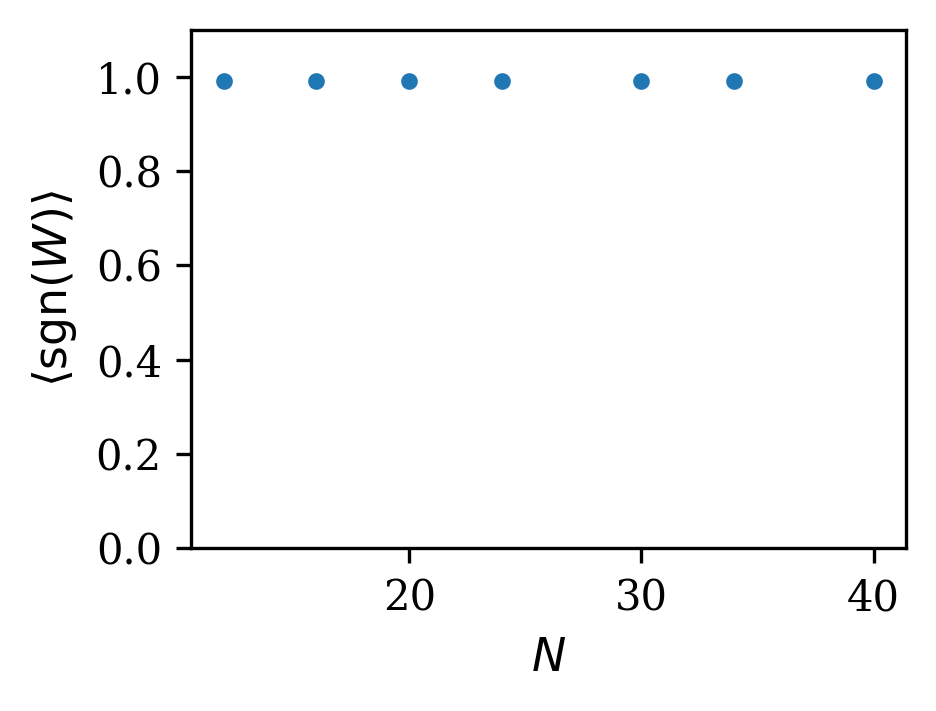}
    \caption{$\langle \text{sgn} \rangle$ vs. $N$ for $N_d=1$ at $\beta=0.1$. The $\langle q \rangle_{\text{QMC}}$, even though in this case, is increases as a function fo $N$ (linearly), it remains less than $3$ all of the relevant system size simulations.}
    \label{fig:avgsgn_vs_N_Nd=1_beta=0.1}
\end{figure}

For slightly larger $\beta \sim 0.5$, let us observe how $\langle q \rangle_{\text{QMC}}$ scales as a function of $N$ in the figure below.
\begin{figure}[H]
    \centering
    \includegraphics[width=6cm]{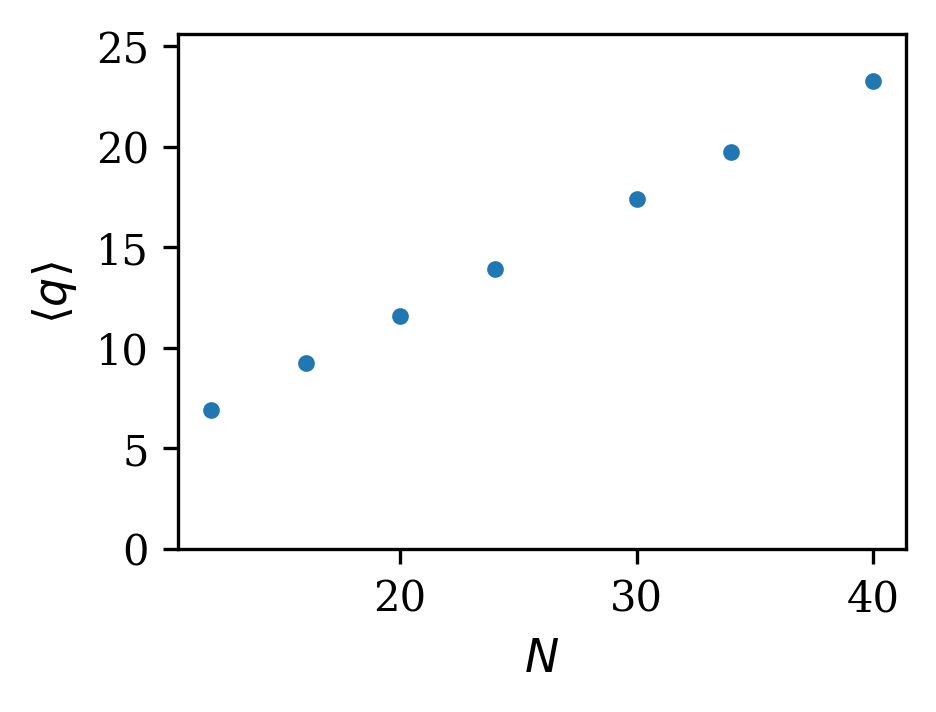}
    \caption{$\langle q \rangle_{\text{QMC}}$ vs. $N$ for $N_d=1$ at $\beta=0.5$. The $\langle q \rangle_{\text{QMC}}$ increases as a function for $N$ (linearly), however, it exceeds the length of the smallest fundamental cycles. This means that fundamental cycles of large length, and concatenation of smaller fundamental cycles (of average length $\langle q \rangle$) contribute to the partition function.}
    \label{fig:avgq_vs_N_Nd=1_beta=0.5}
\end{figure}
For this case, as explained in the cation of the figure above, we expect that fundamental cycles that are comprised of strings of $D_{ij}X_i X_j$ that form a closed loop around the physical lattice of length up to $\langle q \rangle_{\text{QMC}}$ (closed $\langle q \rangle_{\text{QMC}}$-polygons made out of triangles) contribute significantly to the weight of the partition function. 
Let us, in the mean time, observe the value of $\langle \text{sgn} \rangle$ as a function of $N$ below.

\begin{figure}[H]
    \centering
    \includegraphics[width=6cm]{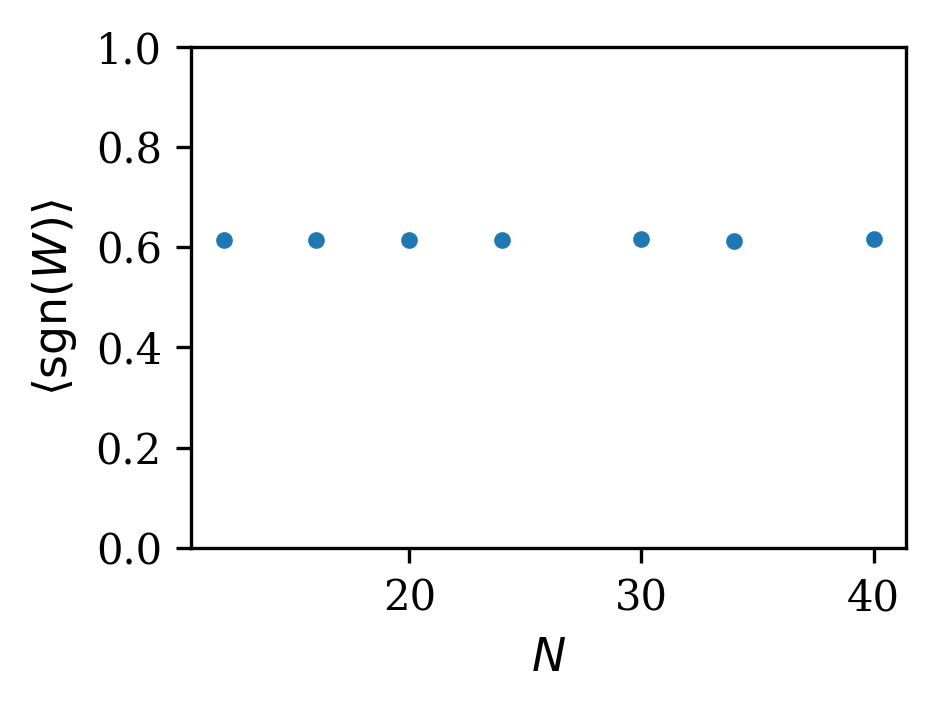}
    \caption{$\langle \text{sgn} \rangle$ vs. $N$ for $N_d=1$ at $\beta=0.5$.}
    \label{fig:avgsgn_vs_N_Nd=1_beta=0.5}
\end{figure}

Counting the number of VGP vs non-VGP fundamental cycles can give us an idea of why the average sign remains constant as a function of $N$, even though $\langle q \rangle_{\text{QMC}}$ increases as a function of $N$, as shown in Fig.~\ref{fig:avgq_vs_N_Nd=1_beta=0.5}. 

For $N$ spins (consider even $N$ for simplicity) on the triangular ladder lattice, there are $\frac{1}{2}(N-2)(N-1)$ closed cycles. Of these, $N-2$ will contain the defect edge, and so the rest will form cycles that produce VGP weights. Since all of the weights of the edges are equal and the magnitude of cycles of equal length will equal, we would only need to keep track of how many VGP vs non-VGP cycles there are to determine the behavior of the average sign. In this instance we can write
\begin{align}
    \langle \text{sgn} \rangle \sim \frac{N_{\text{VGP}} - N_{\text{non-VGP}}}{N_{\text{VGP}} + N_{\text{non-VGP}}} \, .
\end{align}
where $N_{\text{VGP}}$ denotes the number of VGP cycles, and $N_{\text{non-VGP}}$ denotes the number of non-VGP cycles.
The number of VGP cycles $(O(N^2))$ grows faster than the non-VGP ones, and so the average sign will remain constant as a function of $N$.

By adding a few more defect edges we can increase the number of non-VGP cycles, but we cannot change its scaling as a function of $N$, since single isolated defects only produce $O(N)$ number of non-VGP cycles. Below is the average sign plot for $N_d=2$ at $\beta = 0.5$, demonstrating this fact.

\begin{figure}[H]
    \centering
    \includegraphics[width=6cm]{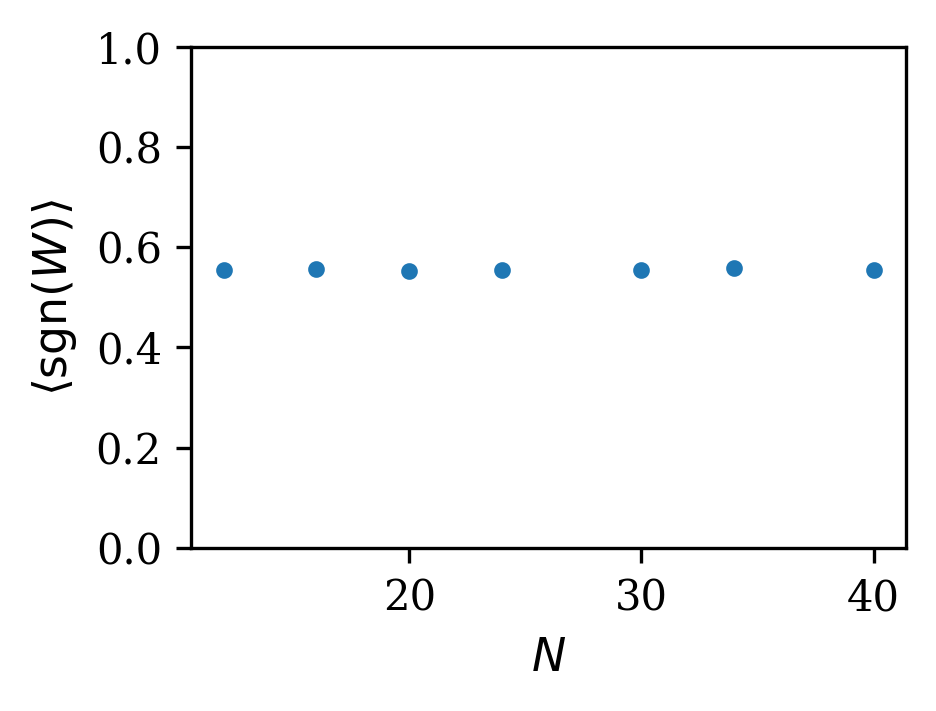}
    \caption{$\langle \text{sgn} \rangle$ vs. $N$ for $N_d=2$ at $\beta=0.5$.}
    \label{fig:avgsgn_vs_N_Nd=2_beta=0.5}
\end{figure}

This toy model was to demonstrate the importance of scaling of VGP fundamental cycles as function of system size and was served to highlight a case with a sign problem, in which the severity of the sign problem does not impede studying the thermodynamic limit (large $N$) of the model. This kind of scaling, however, is quite rare in known problems with a sign problem, because most non-VGP cycles scale similarly to the number of VGP cycles, and so studying large $N$ behavior also becomes a challenge. 
\subsection{Inverse temperature}
\label{subsec:scaling_beta}

While scaling with system size $N$ captures one aspect of the severity of the sign problem, the dependence on the inverse temperature $\beta$ is far more unforgiving. In many cases, the severity with $\beta$ manifests independently of, and often more dramatically than, its scaling with $N$. Understanding this distinction is essential for quantifying how quickly the average sign decays in practical QMC simulations.

We will use the off-diagonal absolutes function $f_\eta(H)$ to highlight the severity of the sign problem with increasing $\beta$. Let ${\lambda_k}$ denote the eigenvalues of $|H_{\text{off}}|$, and ${\lambda'k}$ those of $-H_{\text{off}}$. Then the off-diagonal absolutes function can be expressed as
\begin{equation}
f_{\eta}(H) = \sum_k e^{\eta \lambda_k} - \sum_k e^{\eta \lambda'_k}.
\end{equation}

As noted earlier in the definitions, the parameter $\eta$ in $f_\eta(H)$ plays the same role as the inverse temperature $\beta$ in QMC. The sensitivity of $f_\eta(H)$ to $\eta$ is exponential: $f_\eta(H) = O(\exp(\eta))$. This means that even if the eigenvalues ${\lambda_k}$ and ${\lambda_k'}$, corresponding respectively to the element-wise absolute value of $H_{\text{off}}$ and to $-H_{\text{off}}$, are extremely close but slightly different, $f_\eta(H)$ will still grow exponentially with $\eta$. Consequently, if the average sign is less than $1$ for any $\beta$, it will decay exponentially as $\beta \rightarrow \infty$. 

From the cycle-counting perspective, at very large $\beta$, we have $\langle q \rangle_{\text{QMC}} \gg N$. In this regime, the presence of even a single non-VGP fundamental cycle means any VGP fundamental cycle can be combined with it to form a non-VGP cycle. As $\beta$ increases from zero, we move from a regime where only length-2 fundamental cycles contribute to the partition function, to one where cycles far longer than $N$ dominate. Since $N$ is fixed, growing $\beta$ inflates $\langle q \rangle_{\text{QMC}}$ well beyond the fundamental cycle scale, allowing construction of large non-fundamental cycles purely from fundamental ones. Here, non-VGP and VGP cycles become equally numerous, because each VGP fundamental cycle can pair with a non-VGP fundamental to yield a non-VGP non-fundamental. In the limit $\langle q \rangle_{\text{QMC}} \to \infty$, the ratio of VGP to non-VGP cycles approaches one, hence $\langle \text{sgn} \rangle \rightarrow 0$.

Moreover, since the gap between the ground state and first excited state of many quantum systems scales with $N$, Theorem~\ref{thm:fVGP_hard} implies that deciding whether ground state properties are QMC-simulable can itself be exponentially hard. In particular, determining whether $f_{\text{VGP}}(H) = 0$ for all fundamental cycles of length $O(N)$ can be an exponentially difficult task.
However, geometrically-local periodic Hamiltonians possess structural regularities that can be exploited to test whether $f_{\text{VGP}}(H) = 0$ efficiently. In section~\ref{subsec:VGP_frustrated}, we show that for $2$-local Hamiltonians on a `geometrically frustrated' lattice \footnote{We call a 2D lattice with any $2\times2$ plaquette having all-to-all connectivity, like Fig.\ref{fig:square_lattice_blocks}, a `geometrically frustrated' lattice.}, this determination is in fact tractable.
\section{VGP or stoquastic?}
\label{sec:VGPorStoq}
A Hamiltonian exhibiting a Vanishing Geometric Phase (VGP) can always be transformed into a stoquastic form via a suitable diagonal unitary transformation $\Phi$ \cite{Hen_2021}. In this section, we present several illustrative examples of non-stoquastic Hamiltonians that are nonetheless VGP, demonstrating cases where a simple diagonal unitary $\Phi$ can be explicitly constructed. We also highlight examples where identifying such a transformation is significantly more challenging, requiring finding $k$-body diagonal unitaries ($k>2$), despite the absence of a sign problem. These examples underscore the limitations of stoquasticity as a diagnostic criterion for QMC simulability. To further support this perspective, we provide results for alternative diagnostics based on the VGP framework, focusing in particular on $2$-local Hamiltonians defined on geometrically frustrated lattices: canonical settings where conventional QMC methods are notoriously hindered by the sign problem.

\subsection{Easily `stoquastized' VGP}
Here, we provide a couple of Hamiltonians, that are verifiably VGP, using criterion provided in Eq.\eqref{eq:D_argcheck}, and are simply transformable to stoquastic form, via a diagonal unitary $\Phi$.
The following two Hamiltonians are defined on a square lattice
\begin{align}
    H_1 &= \sum_{\langle i j \rangle} J_{ij} (\mathds{1} + Z_j)X_i \qquad \qquad (J_{ij} \in \mathds{R} ) \label{eq:H1}\\
    H_2 &= \sum_i J_i \big(\Pi_{\langle j,i\rangle}Z_j\big) X_i \qquad \quad (J_i \in \mathds{R}) \label{eq:H2} \, .
\end{align}
Since, the set of permutations in the PMR form of these Hamiltonians have no strings of identity equivalent permutations that do not come in odd powers, the Hamiltonians can be verified to be VGP via Theorem \ref{thm:DX_VGP} (Eq. \eqref{eq:D_argcheck}). 

We will explain below how the Hamiltonians above can be transformed into stoquastic Hamiltonians via two-body Clifford and Pauli operators. However, we will provide an instance in the end of this section for which there are only non-Clifford diagonal unitaries that can transform the VGP Hamiltonian to a stoquastic one.

For $H_1$, one can consider the following diagonal unitary
\begin{align}
    \Phi_1 = \left(\Pi_{i=1}^N (Z_i)^{\delta_i}\right) \Pi_{\langle i , j \rangle := e_{ij} \in E} CZ_{e_{ij}} \, ,
\end{align}
where $\delta_i = \left(1 + (-1)^{\sum_{j \text{ s.t. } \langle j , i \rangle := e_{ij} \in E} J_{ij}}\right)$. In other words, $\Phi_1$ applies controlled-$Z$ rotation to every edge of the graph (with free choice on which spins are control and target). This will transform $H_2$ into
\begin{align}
    \sum_{i=1}^N \tilde{J}_i X_i \, ,
\end{align}
where $\tilde{J}_i = 2\sum_{j \text{ s.t. } \langle j , i\rangle := e_{ij}\in E} J_{ij}$. Then applying a $Z$-transformation on $X_i$ will take $X_i \mapsto - X_i$. Doing this on spins that have $\tilde{J}_i > 0$ will ensure that the resulting Hamiltonian is stoquastic. 
Similar procedure for $H_2$ can be used to transform it into a stoquastic Hamiltonian.

\subsection{Hard to `stoquastize' VGP}

Here we provide an example for which we can easily verify that the Hamiltonian is VGP, but finding a transformation into a stoquastic one involves a search over $5$-body diagonal unitaries, which is extremely difficult to do.
Consider another instance of a Hamiltonian, $H_{\text{Hard}}$, fitting the criteria of Theorem \ref{thm:DX_VGP}, with $H_{\text{Hard}} = \sum_{i=1}^N D_i X_i$ on a square lattice, with $D_i$ of the form
\begin{align}
    D_i := a^{(i)}\mathds{1} + i b^{(i)} Z_i + c^{(i)} \Pi_{\langle j , i\rangle} Z_j \, ,
\end{align}
with $a$, $b$, and $c$ being real parameters. Note that the last diagonal term is a product of Pauli-$Z$ matrices over the $4$ neighbors of the spin $i$.
Since, the sign of the first two coefficients cannot be altered by any other operator in the Hamiltonian, we only consider what happens to the third term. The weights corresponding to $D_i X_i$ for any given even length $2q$ cycle, will be $\gamma^{(i)} \pm c^{(i)}$, where we have defined $\gamma := a^{(i)} \pm i b^{(i)}$. Since, $D_i X_i$ has to appear exactly twice in a given fundamental cycle, the total contribution to the weight could be of two forms
\begin{align}
    d^{(i)}_1 &= \big(\gamma^{(i)} + c^{(i)}\big) \big((\gamma^{(i)})^* + c^{(i)}\big) \notag \\
    &\qquad \qquad \qquad = \big(a^{(i)} + c^{(i)}\big)^2 + (b^{(i)})^2 \geq 0 \\
    d^{(i)}_2 &= \big(\gamma^{(i)} + c^{(i)}\big) \big((\gamma^{(i)})^* - c^{(i)}\big) \, .
\end{align}
The first weight is clearly positive for all values of $a^{(i)}$, $b^{(i)}$, and $c^{(i)}$. In order to ensure that a pair of $D_i X_i$ will not create a sign-problematic weight, we also have to ensure that $d^{(i)}_2 \geq 0$. We can enforce this condition, by picking $c^{(i)}$ such that
\begin{align}
 \Big(c^{(i)} + b^{(i)}\Big)^2 \leq (a^{(i)})^2 + 2(b^{(i)})^2\, .
\end{align} 

In this instance, however, the previously suggested controlled-Z transformations for $H_1$ and $H_2$ would no longer work, and one has to potentially search over a group of $5$-body diagonal unitaries in order to transform $H_{\text{Hard}}$ or other similar Hamiltonians with non-local diagonals, into stoquastic form. 

\subsection{VGP on near-neghboring lattices}
\label{subsec:VGP_frustrated}
Finally, to provide a full scope of the utility of VGP, we also consider a geometrically frustrated lattice structure of the form of the minimal periodic block in FIG. \ref{fig:square_lattice_blocks} put on a one-dimensional lattice. Consider the antiferromagnetic Heisenberg model that is of the form
\begin{align}
    H_{\text{Heis.}} &= - \sum_{\langle i,j\rangle \in G_{\Delta_f}} \vec{S}_i \cdot \vec{S}_j \notag \\
    &= - \sum_{\langle i,j\rangle \in G_{\Delta_f}} \Big( Z_i Z_j + (1 - Z_i Z_j)X_i X_j \Big)\, ,
    \label{eq:heis_model}
\end{align}
where we have used $G_{\Delta_f}$ to denote a geometrically frustrated $2$-D lattice, unlike the quasi-1D ladder lattice considered in \ref{sec:VGPscaling_signseverity}.
The off-diagonal parts of the Hamiltonian is
\begin{align}
    \label{eq:heis_mod}
    H^{\text{(off)}}_{\text{Heis.}}&= \sum_{\langle i,j\rangle \in G_\Delta} D_{ij} X_i X_j \, , \\
    D_{ij} &= (Z_iZ_j - \mathds{1}) \, .
\end{align}

It turns out that $H_{\text{Heis.}}$ is still VGP if we replace $D_{ij}$ by $\tilde{D}_{ij}(\vec{h}^{(ij)})$ for parameters $(\vec{h^{(ij)}})_k \in \mathds{R}$, satisfying $h^{(ij)}_0 \in \mathds{R}^{+}$ (positive and real), given explicitly by
\begin{align}
    \tilde{D}_{ij}(\vec{h}^{(ij)}) &= h^{(ij)}_0(Z_i Z_j - \mathds{1}) + i h^{(ij)}_1 (Z_i + Z_j) \, .
\end{align}
We show that $h^{(ij)}_1$ can be any random real number for $D_{ij}$. As long as they are \emph{of the same sign} for all $D_{ij}$ across the entire lattice, the Hamiltonian given in Eq.\eqref{eq:model_highlight} is VGP. 
In fact, the diagonal part of the Heisenberg Hamiltonian in Eq. (\ref{eq:heis_model}) could be replaced with any random diagonal matrix and this would not affect whether the Hamiltonian is sign problem free.
We provide a proof for the following Hamiltonian
\begin{align}
    \tilde{H}_{\text{Heis.}} &= D_0 + \sum_{\langle i,j\rangle \in G_\Delta} \tilde{D}_{ij}(\vec{h}^{(ij)}) X_i X_j \, ,
    \label{eq:model_highlight}
\end{align}
being VGP in Appendix \ref{proof:modHheis}. This actually is a special case for a Hamiltonian satisfying the condition laid out in Theorem \ref{thm:2localVGP} below. We should state that, it is possible to obtain $\tilde{H}_{\text{Heis.}}$ from a stoquastic Hamiltonian, by applying a series of $R_{Z_i}$ rotations on each spin. We are not aware of a an easier method to convert this model into a stoquastic one through Clifford rotations. In fact, a set of Clifford rotations may not exist for a random choice of a set of $h^{(ij)}_0$ and $h^{(ij)}_1 $, since Cliffords do not cover the entire space of unitary rotations.

We will end this section by providing a proof for a condition under which a $2$-local Hamiltonian on a geometrically frustrated lattice cannot be VGP, and lay out a conjecture on the efficiency of determining whether a $2$-local Hamiltonian with translational invariance is VGP. 
\begin{theorem}[$2$-local geometrically frustrated VGP condition]
    Consider a Hamiltonian on a geometrically frustrated triangular lattice written in PMR form $H = D_0 + \sum_{\langle i,j\rangle \in G_\Delta} D_{ij} X_i X_j$, with $D_{ij}$ being $2$-local diagonals with support on neighboring spins of the form
    \begin{align}
        D_{ij} := h^{(ij)}_0 + i (h^{(ij)}_1 Z_i + h^{(ij)}_2 Z_j) + h^{(ij)}_3 Z_i Z_j \, .
    \end{align}
    $H$ cannot be VGP if $h^{(ij)}_1 \neq h^{(ij)}_2$ for \textbf{all} $D_{ij}$.
    \label{thm:2localVGP}
\end{theorem}

We will outline the proof in Appendix~\ref{proof:2localVGP}. Our analysis in the proof also sets conditions on what values of $h_0$ and $h_3$, for each diagonal term, will satisfy the VGP condition.

The model given in Eq.~\eqref{eq:model_highlight} highlights the power of VGP as a metric for the QMC sign problem over stoquasticity: stoquasticity only says if sign problem is present up to a diagonal unitary. However, general diagonal unitaries are difficult to find, so transforming a Hamiltonian that is already sign-problem-free but a diagonal unitary away from being stoquastic will mistakenly get categorized as sign problematic if stoquasticity is used as a diagnostic test for the QMC sign problem.

\begin{conjecture}[Local VGP Implies Global VGP for Periodic 2-Local Hamiltonians]
\label{conj:2local_periodic_VGP}
Let $H$ be a Hamiltonian defined on a $d$-dimensional periodic lattice $\Lambda \subset \mathbb{Z}^d$, equipped with a discrete translational symmetry group $\mathcal{T} = \langle T_1, \dots, T_d \rangle$, where $T_\mu$ denotes translation along the $\mu$-th lattice direction.

Suppose $H$ admits a Permutation Matrix Representation (PMR) of the form
\[
H = \sum_{i \in \mathcal{I}} D_i P_i,
\]
where each term $D_i P_i$ is a $2$-local Hermitian operator acting nontrivially only on a pair of spins connected by an edge of the lattice.

Assume the following:

\begin{enumerate}
    \item \textbf{Translational Periodicity:} There exists a finite \emph{unit cell} $U \subset \Lambda$ such that for all $x \in \mathcal{T}$, the PMR terms obey
    \[
    T_x D_i T_x^{-1} = D_j \quad \text{and} \quad T_x P_i T_x^{-1} = P_j
    \]
    for some $j \in \mathcal{I}_U$, meaning the diagonal and permutation operators repeat under lattice translations.

    \item \textbf{Unit Cell Restriction:} Let $H_{\mathrm{unit}}$ denote the restriction of $H$ to a single unit cell $U$, including all 2-local PMR terms whose supports lie entirely within $U$.
\end{enumerate}

Then, if
\[
f_{\mathrm{VGP}}(H_{\mathrm{unit}}) = 0,
\]
it follows that the full Hamiltonian $H$ defined on any periodic extension of the lattice also satisfies
\[
f_{\mathrm{VGP}}(H) = 0.
\]
That is, the vanishing geometric phase property of the full system follows from that of a single unit cell, provided the Hamiltonian is 2-local and periodic.
\end{conjecture}

\begin{figure}[H]
    \centering
    \includegraphics[width=5cm]{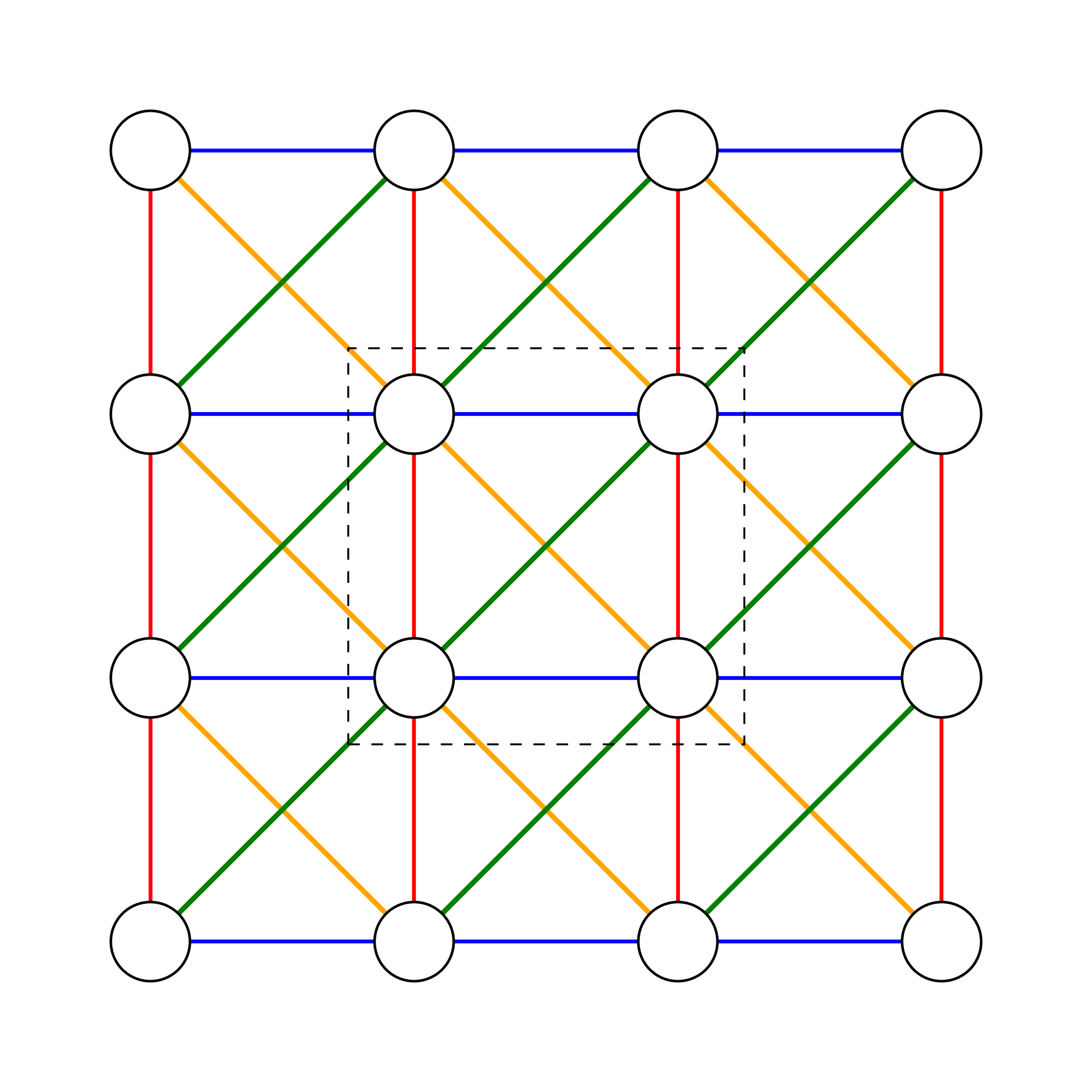}
    \caption{Unit block for the large rectangular lattice with $2$-local interaction on neighboring spins.}
    \label{fig:square_lattice_blocks}
\end{figure}

This conjecture is supported by heuristical arguments, provided in Appendix \ref{app:2body_periodic_VGP}, as well as numerical investigations, available on GitHub~\cite{OurGitHub}. The goal of the numerical experiments will be to find a set of VGP Hamiltonians through optimizing parameters of a $2$-local Hamiltonian on the unit cell of a geometrically frustrated lattice, verify that VGPness holds with the same parameters over a periodic construction of the larger lattice, using the unit cell parameters. We then search over hundreds of such instances and verify that all such instances satisfy our conjecture. We describe the setting and the procedure for our numerical setup below:

\begin{enumerate}
    \item First we start with a $2 \times 2$ lattice structure.
    \item We set the permutation terms to be of the form $X_i X_{i+1}$.
    \item We parametrize the Hamiltonian via the four unique Diagonal operators of the unit block of the form
    \begin{align}
        D^{(ij)} = h^{(ij)}_0 \mathds{1} + i h^{(ij)}_1 (Z_i + Z_j) +  h^{(ij)}_2 Z_i Z_j \, ,
    \end{align}
    for $(i,j) \in E(\Lambda_U)$ (the edges of the unit lattice), and optimize in order to find a unit block with $f_\eta(H_{\text{unit}}) = 0$.
    
    \item To generate other similar VGP Hamilaonians, we perform diagonal unitaries: once we find a set of 12 parameters, ($3$ for each edge), we perform $R_Z(\alpha)$ for a random $\alpha$ on all spins, in order to retain the periodicity. 
    
    \item We can compute an average over 200 instances, and verify that for numerous attempts at finding different optimized unit block Hamiltonians, for which $f_\eta(H_{\text{unit}}) = f_{\text{VGP}}(H_{\text{unit}}) = 0$, the translationally generated Hamiltonian on the larger lattices is also VGP.
\end{enumerate}

Please refer to our GitHub \cite{OurGitHub} for the full code that performs this computation.

\section{Sign problem in random basis}\label{sec:avg_sign_stats}

In this section, we analyze the the sign problem using random unitary transformations, by analyzing expected value of $\mathcal{F}_H$, a metric introduce in \ref{sec:measures}. In summary, we show that random unitaries are devastating for the sign problem by showing the scaling behavior of the expected value of $\mathcal{F}_H$ and the vanishing probability of deviations from the mean. This of course, is nothing surprising, as computing the ground state of random Hamiltonians is known to be a QMA-hard problem~\cite{Jordan_2009}. We provide further numerical evidence that even a perfectly sign problem free Hamiltonian, is expected to suffer from a severe sign problem under random unitary transformations. 

In order to carry out our analysis, we utilize a few well-established theoretical results from random matrix theory, including Lévy’s lemma that, said informally, demonstrates that Lipschitz functions $f$ on high dimensional spaces are tightly concentrated about $\E[f]$.
Lévy’s lemma states
\begin{lemma}[L\'evy's Lemma]
Let $S^{n}$ be the unit sphere in $\mathbb{R}^{n+1}$ with the uniform (Haar) measure $\mu$. 
If $f : S^{n} \to \mathbb{R}$ is $L$-Lipschitz, i.e.
\[
|f(x) - f(y)| \le L \,\|x - y\|_2 \quad \forall\, x,y \in S^{n},
\]
then for any $\varepsilon > 0$,
\[
\mu\left( \left\{ x \in S^{n} : |f(x) - \mathrm{Med}(f)| \ge \varepsilon \right\} \right)
\le 2 \exp\left( -\frac{c\,n\,\varepsilon^{2}}{L^{2}} \right),
\]
where $\mathrm{Med}(f)$ is the median of $f$ and $c>0$ is an absolute constant.
\end{lemma}

We apply Levy's lemma, as well as Jensen's inequality~\cite{Durrett_2019} to show that, with high probability, the sum of the magnitude of configuration weights $\mathcal{F}_H$ for a physical Hamiltonian $H$’s partition function is sharply concentrated about its mean when averaged over all unitary transformations. We then find an analytical formula for the expectation of $\mathcal{F}_H$ averaged over all possible unitaries. We arrive at a similar result for random Clifford unitaries as well. These result not only state that a Hamiltonian, rotated in a random basis, will have devastatingly severe sign problem, alluding to hardness of simulating the ground state properties of random Hamiltonians, but also deem random searches over the entire Clifford space to find sign curing/mitigating basis a hopeless task.

As a note, in this section we use $\lambda_j(\, \cdot \,)$ to denote the $j$th largest eigenvalue of a given Hermitian matrix.
Now we proceed with the main results of this section. Recall the definitions $\mathcal{M}(X) = |X_{\text{off}}| - \text{diag}(X)$ and $\mathcal{F}_X(U) = \tr[e^{\mathcal{M}(UXU^\dagger)}]$. Then $\mathcal{F}_H$ exhibits the concentration of measure phenomenon:

\begin{theorem}\label{thm:globalconcentration}
    Suppose $H \in \text{Herm}(\mathds{D})$ is an arbitrary Hamiltonian acting on $N$ qubits with dimension $\mathds{D} := 2^N$. Assume that the spectral norm of $H$ satisfies $\|H\|_{\text{op}} = O(N)$. Fix $\alpha > 0$.
    Then there exists $\epsilon(N) = o(1)$ such that if $U \in U(N)$ is drawn randomly from the Haar measure, we have
    \begin{align*}
        \text{Pr}\left( |\mathcal{F}_H(U) - \E[\mathcal{F}_H]| \geq \epsilon \cdot \E[\mathcal{F}_H] \right) \leq \frac{1}{\mathds{D}^\alpha}
    \end{align*}
\end{theorem}
For the proof of the above theorem, we need the following lemma:
\begin{lemma}\label{lem:key}
     Suppose $X \in \text{Herm}(\mathds{D})$. Then the map $\mathcal{L}_X := \ln \mathcal{F}_X$ is $L$-Lipschitz with respect to the Frobenius norm and constant $L = 2\|X\|_{\text{op}}$.
\end{lemma}
See Appendix \ref{app:lemkey_proof} for the proof of this lemma.
Now we may proceed with the proof of Theorem~\ref{thm:globalconcentration}.
\begin{proof}
    By Lemma~\ref{lem:key}, the map $\mathcal{L}_H$ is $L$-Lipschitz with constant $L = 2\|H\|_{\text{op}}$. Set
    \begin{align*}
        \delta = 2 \|H\|_{\text{op}} \sqrt{\frac{9\pi^3 \ln(2\mathds{D}^\alpha)}{\mathds{D}^2}}.
    \end{align*}
    Let $U\in U(\mathds{D})$ be drawn from the Haar measure. Then by Lévy's lemma~\cite{Ledoux_2001}, 
    \begin{align*}
        \tag{2}
        \text{Pr}\left(|\mathcal{L}_H(U) - \E[\mathcal{L}_H]| \geq \delta\right) \, &\leq \, 2\exp\left( -\frac{\mathds{D}^2 \delta^2}{9\pi^3L^2} \right)\\
        &= 2\exp\left( -\frac{\mathds{D}^2 \delta^2}{36\pi^3\|H\|_{\text{op}}^2} \right)\\
        &= \frac{1}{\mathds{D}^\alpha}.
    \end{align*}
    Now, for the remainder of the proof, consider when $| \mathcal{L}_H(U) - \E[\mathcal{L}_H] | \leq \delta$. In this case,
    \begin{align*}
         \mathcal{F}_H(U) \in \left[ e^{\E[\mathcal{L}_H] - \delta},\, e^{\E[\mathcal{L}_H] + \delta} \right].
    \end{align*}
    By Jensen's inequality, we have $e^{\E[\mathcal{L}_H]} \leq \E[\mathcal{F}_H]$, so
    \begin{align*}
        \frac{\mathcal{F}_H(U)}{\E[\mathcal{F}_H]} \leq \frac{e^{\E[\mathcal{L}_H] + \delta}}{e^{\E[\mathcal{L}_H]}}
        = e^{\delta}.
    \end{align*}
    Then by Taylor expanding,
    \begin{align*}
        \left| \frac{\mathcal{F}_H(U) - \E[\mathcal{F}_H]}{\E[\mathcal{F}_H]} \right| \leq e^{\delta} - 1 \leq \delta + O(\delta^2) =: \epsilon(n).
    \end{align*}
    By assumption, we know that $\|H\|_{\text{op}} = O(N)$, so $\delta = O\left( n^{3/2} \cdot 2^{-n} \right)$.
    Then $\epsilon(n) = o(1)$ and 
    \begin{align*}
        \text{Pr}\left( |\mathcal{F}_H(U) - \E[\mathcal{F}_H]| \geq \epsilon \cdot \E[\mathcal{F}_H] \right) \leq \frac{1}{\mathds{D}^\alpha}.
    \end{align*}

\end{proof}
Next, we derive an explicit formula for $\E[\mathcal{F}_H]$ with vanishingly small relative error, valid for any Hermitian $H$. Through numerical simulations on various physical Hamiltonians, we demonstrate the formula's low relative error against empirical measurements of $\E[\mathcal{F}_H]$.

\begin{theorem}\label{thm:noise_formula}
    Let $H$ be a Hamiltonian acting on a Hilbert space of dimension $\mathds{D}$.  
    Denote by $\lambda_1(H') \ge \lambda_2(H') \ge \cdots \ge \lambda_{\mathds{D}}(H')$ the eigenvalues of a Hermitian matrix $H'$.  
    Define
    \begin{align*}
    \sigma^2 &:= \Bigl(2-\tfrac{\pi}{2}\Bigr)\frac{\|H\|_F^2}{2\mathds{D}^2},\\
    \mu &:= \frac{\|H\|_F}{2\mathds{D}}\sqrt{\pi},\\
    \lambda_\star
      &:= -\mu + \mu \mathds{D} + \frac{\sigma^2}{\mu}.
    \end{align*}
    When $\mathds{D}$ is sufficiently large, the lesser $\mathds{D}-1$ eigenvalues of $H' = \mathcal{M}(UHU^\dagger)$
    converge to the Wigner semicircle distribution on $[-2\sigma\sqrt{\mathds{D}},2\sigma\sqrt{\mathds{D}}]$, a universal distribution arising in random matrix theory~\cite{Tao_2012}, while  
    \[
        \lambda_{1}(H') = \lambda_\star + O(\mathds{D}^{-1}).
    \]
    In all, the expectation value of $\mathcal F_H$ is given by
    \begin{align}
    \mathbb{E}[\mathcal F_H]
      &= e^{\lambda_\star}\left(1 
        +
        (\mathds{D}-1)\frac{I_1\bigl(2\sigma\sqrt \mathds{D}\bigr)}{e^{\lambda_\star}\sigma\sqrt \mathds{D}}
        +
        \frac{\epsilon(H)}{e^{\lambda_\star}}\right),
    \end{align}
    where $I_1$ is the modified Bessel function of the first kind, $\sigma$ and $\mu$ are given above in terms of $H$, and the error term satisfies  
    \[
        \frac{|\epsilon(H)|}{e^{\lambda_\star}} = O\bigl(\mathds{D}^{-1}\bigr).
    \]
\end{theorem}

Refer to Appendix \ref{app:noiseformula_proof} for the proof.

We perform a numerical analysis on random Heisenberg models to demonstrate that the derived formula for $\E[\mathcal{F}_H]$ indeed has vanishingly small relative error. In this analysis, we take the average of $\mathcal{F}_H$ across 30 random Haar-drawn unitaries applied to each $H$, where $H$ is a random $n$-body Hamiltonian with nearest-neighbor Heisenberg interactions conjugated by a random Haar-drawn unitary $U \in U(\mathds{D})$. 
\begin{figure}[H]
    \centering
    \includegraphics[width=8cm]{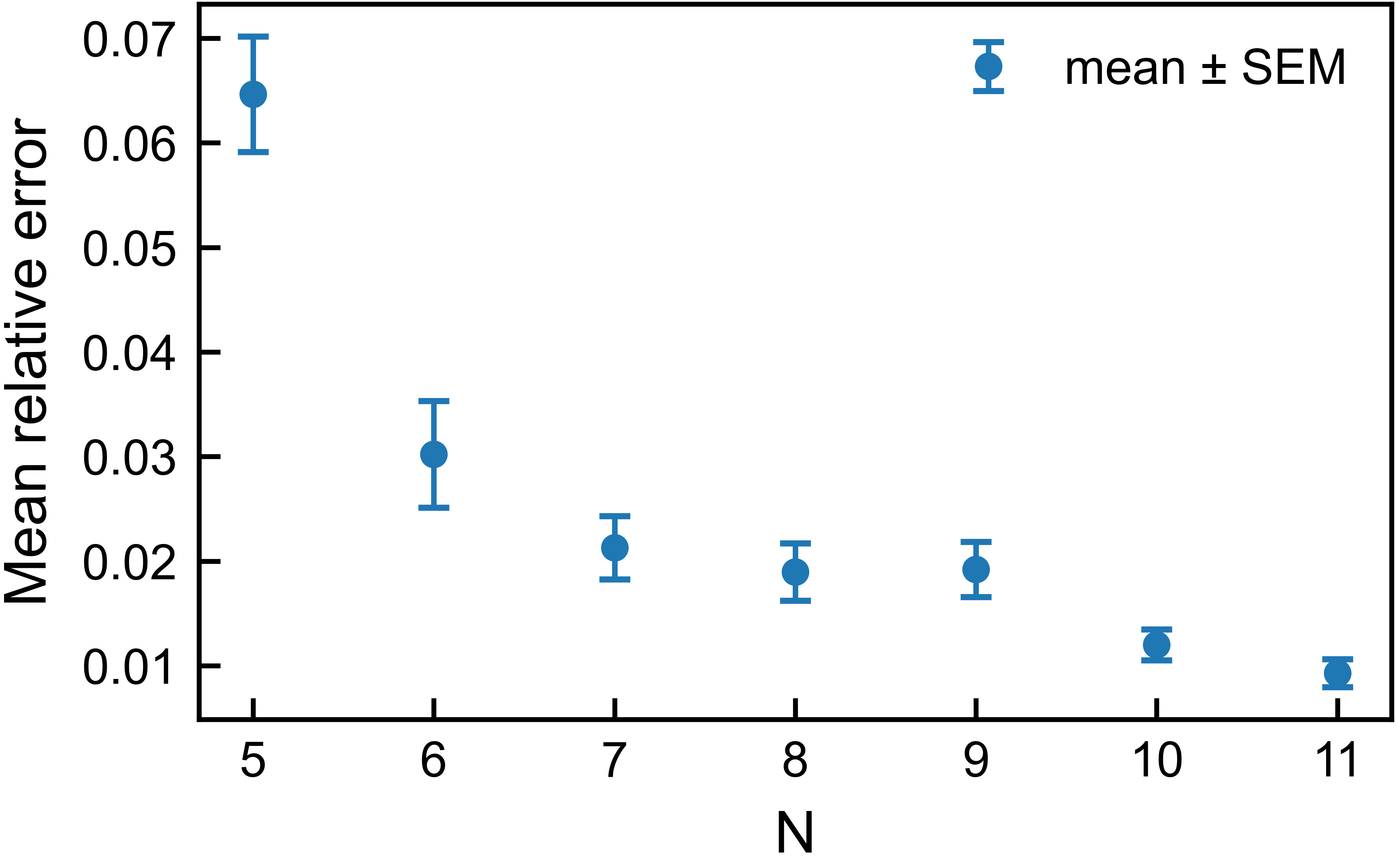}
    \caption{Scaling of relative error of Theorem~\ref{thm:noise_formula} formula for $\E[\mathcal{F}_H]$ vs empirical $\E[\mathcal{F}_H]$}
    \label{fig:exp_f_graph}
\end{figure}

Recall that the average QMC sign of $H$ in a given basis is $\frac{Z(H)}{\mathcal{F}_H(\mathds{1})}$. Then Theorems~\ref{thm:globalconcentration} and~\ref{thm:noise_formula} formally demonstrate that the vast majority of basis transformations $U$ will result in an extremely low average sign.

The following result allows us to study the average sign behavior of the majority of random spin-$\frac{1}{2}$ Hamiltonians and we conjecture that it may also serve as a tight estimate for the expectation of $\mathcal{F}_H$ with domain restricted to the Clifford group.
\begin{theorem}[Random Pauli strings]
\label{thm:exp_f_rand_cliff}
    Let $(\tilde{P}_1,\ldots,\tilde{P}_M) \in (\mathcal{P}_N \setminus \{\mathds{1}\})^M$ be a tuple of non-identity Pauli operators each acting on $N$ qubits and each drawn independently at random from $\mathcal{P}_N \setminus \{\mathds{1}\}$. Define the spin-1/2 Hamiltonian
    \begin{align*}
        H = \sum_{j=1}^M c_j\tilde{P}_j = \sum_{j=1}^M c_jD_jP_j
    \end{align*}
    where $c_j \in \mathbb{R}$ are fixed, $\tilde{P}_j \in \mathcal{P}_N$, $D_j$ is diagonal, $P_j$ is a permutation matrix, $M = O(\text{poly}(N))$ and each $c_j = O(1)$ with $N$. Set $c_{\text{tot}} := \sum_{j=1}^M |c_j|$. Then with probability $1 - O(M^22^{-N})$, $H$ belongs to a class $\mathcal{H}$ of spin-1/2 Hamiltonians such that
    \begin{align*}
    \E[\mathcal{F}_H(\mathds{1}) \, \vert \, H \in \mathcal{H}] &= e^{c_{\text{tot}}} \notag \\
    &+ \left( \mathds{D} - 1 \right) \prod_{j=1}^M \cosh(|c_j|) \left[ 1 + O(M^22^{-N}) \right].
    \end{align*}
\end{theorem}
Refer to Appendix \ref{app:randcliff_proof} for the proof.

\section{Discussion}

One of the central observations of this work is that VGP-based classification extends beyond the stoquastic/non-stoquastic dichotomy: there exist Hamiltonians that are easily recognized as sign-problem-free through the VGP criterion, yet whose stoquasticity is highly nontrivial to certify. This distinction emphasizes that VGP captures a broader set of simulable instances, making it an attractive alternative or complement to stoquasticity-based approaches.

From a complexity-theoretic perspective, we have shown that determining whether a general Hamiltonian is VGP is hard in the worst case. Nevertheless, our analysis uncovers certain sparse Hamiltonians and other subclasses with periodic structure, notably 2-local periodic Hamiltonians, for which diagnosing the sign problem is efficient.

We have also introduced a family of VGP-inspired measures that quantify the \emph{severity} of the sign problem. While the exact computation of these measures is generically intractable, they prove to be useful in determining the severity of the sign-problem under general unitary and Clifford transformations.

Collectively, these results unify structural, algorithmic, and statistical viewpoints on the sign problem, offering a conceptual framework that not only highlights the origins of the sign problem but also hints at principled pathways for its mitigation.

\section{Summary and Outlook}
In summary, we have
\begin{itemize}
    \item Examined the utility of VGP criterion as a geometric diagnostic for QMC simulability, grounded in the permutation matrix representation of computational state graphs.
    \item Analyzed the complexity of determining VGP status, identifying both provably hard cases and efficiently checkable subclasses.
    \item Constructed explicit examples of Hamiltonians that are VGP yet difficult to recognize as stoquastic, highlighting the practical advantage of VGP over traditional stoquasticity-based criteria.
    \item Provided theorems and a conjecture on efficiency of determining the sign-problem in sparse and also in periodically structured instances. 
    \item Proposed a set of quantitative measures inspired by VGP, enabling scaling analyses of the average sign under basis transformations.
\end{itemize}

Looking ahead, several research directions naturally emerge:
\begin{enumerate}
    \item \textbf{Algorithmic development for special cases:} Further investigation into 2-local periodic Hamiltonians and related structures may yield polynomial-time algorithms for QMC simulability detection, as well as a rigorous proof of efficient certification of sign problem through local diagnostics.
    \item \textbf{Extensions to fermionic systems:} While our examples focused on spin Hamiltonians, the VGP framework should generalize to fermionic models via Jordan--Wigner or Bravyi--Kitaev mappings, potentially revealing new classes of simulable systems.
    \item \textbf{Integration with mitigation protocols:} VGP diagnostics could inform adaptive basis rotation schemes, guiding unitary transformations toward sign-reducing representations in a structured rather than random manner.
    \item \textbf{Complexity-theoretic boundaries:} Establishing tighter upper and lower bounds on the complexity of VGP determination will clarify its precise relationship to known hardness results for Hamiltonian ground state problems.
\end{enumerate}

We expect that the interplay between geometric phase structure, computational complexity, and QMC performance illuminated here will inspire both theoretical advances and practical algorithms, moving the field closer to a systematic understanding---and eventual control---of the sign problem for tractable instances.

\section{Acknowledgements}

The authors would like to thank Itay Hen, Lev Barash, and Milad Marvian for discussions. 
\bibliographystyle{apsrev4-2}
\bibliography{refs}

\appendix

\onecolumngrid
\section{Proof of Theorem \ref{thm:2localVGP}}
\label{proof:2localVGP}
To prove this theorem, we consider length-$3$ cycles formed by any three of diagonal and permutations on any closed triangular edges of the geometrically frustrated lattice. For the sake of generality, we will consider the edges: $e_1 := (i,j)$ , $e_2 := (j,k)$ and $e_3 := (i,k)$. To write the diagonal weights of such a cycle, we will pick a computational basis state $\ket{z}$, and since the operators only act on neighboring spins, the diagonal weight will only be a function of three spins, namely $z_i$, $z_j$, and $z_k$. These are binary variables that introduce a $\pm$ sign to each of the four variables for each diagonal term. The weight of the cycle $W^{(z_i , z_j , z_k)}_{ijk}$ is thus 
\begin{align}
    W^{(z_i , z_j , z_k)}_{ijk}&:= \bra{z_i z_j z_k} D_{ij} X_i X_j D_{jk} X_j X_k D_{ik} X_i X_k \ket{z_i z_j z_k} \notag \\
    &= (h^{(ij)}_0 + i (h^{(ij)}_1 \bar{z}_{i} + h^{(ij)}_2 \bar{z}_{j}) + h^{(ij)}_3 \bar{z}_i\bar{z}_j ) \notag \\
    &\times (h^{(jk)}_0 + i (h^{(jk)}_1 z_{j} + h^{(jk)}_2 \bar{z}_{k}) + h^{(jk)}_3 z_j\bar{z}_k )(h^{(ik)}_0 + i (h^{(ik)}_1 z_{i} + h^{(ik)}_2 z_{k}) + h^{(ik)}_3 z_iz_k )\, .
    \label{eq:2local_weight}
\end{align}
We emphasize that since $z_l$ is a binary variable, after a corresponding $X_l$ action, we flip it to $\bar{z}_l$. As we can see, $z_l$ and $\bar{z}_l$ for $l \in \{i,j,k\}$ appear in the weight above. 
Furthermore, as mentioned in the review of PMR formalism in section \ref{sec:pmrqmc_rev}, only the real part of the weights contribute to the partition function or any expectation value of physical operators. However, for a Hamiltonian to be VGP (Eqs.~\eqref{eq:VGP_angle} and ~\eqref{eq:fc_weight}) implies that no fundamental cycle should have a complex component, as it could contribute to a negative higher length cycle weight. Expanding the weight above, and enforcing that the imaginary components vanish, we get 
\begin{align}
    \Im(W^{(z_i , z_j , z_k)}_{ijk}) &= (h^{(ij)}_1 \bar{z}_{i} + h^{(ij)}_2 \bar{z}_{j}) (h^{(jk)}_0 + h^{(jk)}_3 z_j\bar{z}_k )(h^{(ik)}_0 + h^{(ik)}_3 z_iz_k ) \notag \\
    &+ (h^{(jk)}_1 z_{j} + h^{(jk)}_2 \bar{z}_{k}) (h^{(ij)}_0 + h^{(ij)}_3 \bar{z}_i\bar{z}_j )(h^{(ik)}_0 + h^{(ik)}_3 z_iz_k ) \notag \\
    &+ (h^{(ik)}_1 z_{i} + h^{(ik)}_2 z_{k})(h^{(ij)}_0 + h^{(ij)}_3 \bar{z}_i\bar{z}_j )(h^{(jk)}_0 + h^{(jk)}_3 z_j\bar{z}_k ) \notag \\
    &-(h^{(ij)}_1 \bar{z}_{i} + h^{(ij)}_2 \bar{z}_{j})(h^{(jk)}_1 z_{j} + h^{(jk)}_2 \bar{z}_{k})(h^{(ik)}_1 z_{i} + h^{(ik)}_2 z_{k}) = 0\, .
    \label{app1:W_im}
\end{align}
We will come back to this equality soon. First, let us also collect the real part and set it to less than or equal to zero (since the real component of a length-$3$ cycle should have a phase of $3\pi \mod 2\pi$, i.e. be a negative real number, according to Eq.~\eqref{eq:VGP_angle})
\begin{align}
    \Re(W^{(z_i , z_j , z_k)}_{ijk}) &=  (h^{(ij)}_0 + h^{(ij)}_3 \bar{z}_i\bar{z}_j )(h^{(jk)}_0 + h^{(jk)}_3 z_j\bar{z}_k )(h^{(jk)}_0 + h^{(jk)}_3 z_jz_k ) \notag \\
    &-(h^{(ij)}_1 \bar{z}_{i} + h^{(ij)}_2 \bar{z}_{j})(h^{(jk)}_1 z_{j} + h^{(jk)}_2 \bar{z}_{k})(h^{(jk)}_0 + h^{(jk)}_3 z_jz_k ) \notag \\
    &-(h^{(ij)}_1 \bar{z}_{i} + h^{(ij)}_2 \bar{z}_{j})(h^{(ik)}_1 z_{i} + h^{(ik)}_2 z_{k})(h^{(jk)}_0 + h^{(jk)}_3 z_j\bar{z}_k ) \notag \\
    &-(h^{(jk)}_1 z_{j} + h^{(jk)}_2 \bar{z}_{k})(h^{(ik)}_1 z_{i} + h^{(ik)}_2 z_{k})(h^{(ij)}_0 + h^{(ij)}_3 \bar{z}_i\bar{z}_j ) \leq 0\, . 
    \label{app1:W_re}
\end{align}

For notational convenience, we define
\begin{align}
    \mathcal{I}_{ij} (z_i , z_j) &:= h_1^{(ij)} z_i + h_2^{(ij)} z_j, \\
    \mathcal{R}_{ij} (z_i , z_j) &:= h_0^{(ij)} + h_3^{(ij)} z_i z_j \, ,
\end{align}
and rewrite Eq.~\eqref{app1:W_im} as
\begin{align}
    \mathcal{I}_{ij}(\bar{z}_i, \bar{z}_j)\, \mathcal{R}_{jk}(z_j, \bar{z}_k)\, \mathcal{R}_{ik}(z_i, z_k)
    + \mathcal{I}_{jk}(z_j, \bar{z}_k)\, \mathcal{R}_{ij}(\bar{z}_i, \bar{z}_j)\, \mathcal{R}_{ik}(z_i, z_k) &+ \mathcal{I}_{ik}(z_i, z_k)\, \mathcal{R}_{ij}(\bar{z}_i, \bar{z}_j)\, \mathcal{R}_{jk}(z_j, \bar{z}_k) \notag \\
    &=\mathcal{I}_{ij}(\bar{z}_i, \bar{z}_j)\, \mathcal{I}_{jk}(z_j, \bar{z}_k)\, \mathcal{I}_{ik}(z_i, z_k)\, ,
    \label{app1:W_im_simp}
\end{align}
and the inequality in Eq.~\eqref{app1:W_re} as
\begin{align}
    \mathcal{R}_{ij}(\bar{z}_i, \bar{z}_j)\, \mathcal{R}_{jk}(z_j, &\bar{z}_k)\, \mathcal{R}_{ik}(z_j, z_k) - \mathcal{I}_{ij}(\bar{z}_i, \bar{z}_j)\, \mathcal{I}_{jk}(z_j, \bar{z}_k)\, \mathcal{R}_{ik}(z_j, z_k) \notag \\
    &- \mathcal{I}_{ij}(\bar{z}_i, \bar{z}_j)\, \mathcal{I}_{ik}(z_i, z_k)\, \mathcal{R}_{jk}(z_j, \bar{z}_k) - \mathcal{I}_{jk}(z_j, \bar{z}_k)\, \mathcal{I}_{ik}(z_i, z_k)\, \mathcal{R}_{ij}(\bar{z}_i, \bar{z}_j) \leq 0\, . \label{app1:W_re_simp}
\end{align}

\subsection{Case 1: all single \texorpdfstring{$Z$}{Z} terms have equal coefficients}
Note that if all the single $Z$ coefficient values are equal, i.e. $h_1 = h_2$ for all diagonal operators, then the right hand side of the equation above will be zero, for any choice of $z_i$, $z_j$, and $z_k$, regardless of the specific values for each diagonal operator. This can be explicitly shown by noting \textbf{Observation 2} in Appendix~\ref{proof:modHheis}. In short, if all single-$Z$ weights ($h_1$ and $h_2$) are exactly equal, one can guarantee that at least one $\mathcal{I}$ will vanish for any choice of $z_i$, $z_j$, and $z_k$. This is because there is at least one mis-alignment of spins in the spin configurations of the following three terms that show up on the right hand side of Eq.~\eqref{app1:W_im_simp} 
\begin{align*}
    \mathcal{I}_{ij}(\bar{z}_i, \bar{z}_j)\, ,\\
    \mathcal{I}_{jk}(z_j, \bar{z}_k)\, ,\\
    \mathcal{I}_{ik}(z_i, z_k) \, ,
\end{align*}
enforcing at least one of them to vanish.

Let us pick, without loss of generality, $\mathcal{I}_{ij}(\bar{z}_i, \bar{z}_j) = 0 $ (two instances: $z_i = \bar{z}_j = 1$ or $z_i = \bar{z}_j = -1$). In this case, the equation above, reduces to
\begin{align}
    \mathcal{R}_{ij}(\bar{z}_i, \bar{z}_j)\Big(\mathcal{I}_{jk}(z_j, \bar{z}_k)\, \, \mathcal{R}_{ik}(z_i, z_k) &+ \mathcal{I}_{ik}(z_i, z_k)\, \mathcal{R}_{jk}(z_j, \bar{z}_k) \Big) =0\, .
    \label{app1:W_im_simp2}
\end{align}

Furthermore, the inequality in~\eqref{app1:W_re_simp} becomes
\begin{align}
    \mathcal{R}_{ij}(\bar{z}_i, \bar{z}_j)\, \Big(\mathcal{R}_{jk}(z_j, \bar{z}_k)\, \mathcal{R}_{ik}(z_j, z_k)
    - \mathcal{I}_{jk}(z_j, \bar{z}_k)\, \mathcal{I}_{ik}(z_i, z_k)\, \Big)
    \leq 0\, .
    \label{app1:W_re_simp2}
\end{align}
The two conditions obtained above are consistent as the case where $\mathcal{R}_{ij}(\bar{z}_i, \bar{z}_j) = 0$ will obviously satisfy both equalities, and the other case implies (if $\mathcal{I}_{jk}(z_j, \bar{z}_k) \neq 0$)
\begin{align}
    \mathcal{R}_{ik}(z_i, z_k) = - \frac{\mathcal{I}_{ik}(z_i, z_k)\, \mathcal{R}_{jk}(z_j, \bar{z}_k)}{\mathcal{I}_{jk}(z_j, \bar{z}_k)} \, . 
\end{align}
Substituting this expression in the inequality in~\eqref{app1:W_re_simp2} will result in
\begin{align}
    \label{app:ineq_RII}
    -\frac{\mathcal{R}_{ij}(\bar{z}_i, \bar{z}_j)\, \mathcal{I}_{ik}(z_i, z_k)}{\mathcal{I}_{jk}(z_j, \bar{z}_k)}\, \Big((\mathcal{R}_{jk}(z_j, \bar{z}_k))^2\, 
    + (\mathcal{I}_{jk}(z_j, \bar{z}_k))^2\, \Big)
    \leq 0 \implies \mathcal{R}_{ij}(\bar{z}_i, \bar{z}_j)\frac{\mathcal{I}_{ik}(z_i, z_k)}{\mathcal{I}_{jk}(z_j, \bar{z}_k)} \geq 0 \, .
\end{align}
Note that this final condition must hold for $4$ different cases: $2$ for the two instances of $z_i = \bar{z}_j$ mentioned above and two other for different instances of $z_k$. This means that the ratio $\frac{\mathcal{I}_{ik}(z_i, z_k)}{\mathcal{I}_{jk}(z_j, \bar{z}_k)}$ can take on both positive and negative values. This is because even when the condition $h_1 = h_2$ holds for the two other diagonal terms on the triangle of geometrically frustrated lattice, there is an instance in which the two coefficients add up to a positive and another instance of $z_i , z_j$, and $z_k$ adding to a negative value, allowing for  $\frac{\mathcal{I}_{ik}(z_i, z_k)}{\mathcal{I}_{jk}(z_j, \bar{z}_k)}$ to be both positive and negative. 
Here we note that we arrive at a contradiction since if $\frac{\mathcal{I}_{ik}(z_i, z_k)}{\mathcal{I}_{jk}(z_j, \bar{z}_k)}$ could take on both positive and negative values, the only way to ensure the last inequality is to assure $\mathcal{R}_{ij}(\bar{z}_i, \bar{z}_j)=0$ when the ratio $\frac{\mathcal{I}_{ik}(z_i, z_k)}{\mathcal{I}_{jk}(z_j, \bar{z}_k)}$ is negative. Thus, any $2$-local Hamiltonian of the form discussed here that does not satisfy this condition would not be VGP, even if all the single-$Z$ terms have equal coefficients.

Thus, we have shown that if all the weights $h_1$ and $h_2$ are equal in magnitude, one can always check if the last inequality in Eq.~\eqref{app:ineq_RII} is satisfied for the given parameters of the diagonal matrices.

Below, we will show that if we have a diagonal term that has unequal single-$Z$ terms, we arrive at infeasible set of inequalities, and no choice of parameters can make the Hamiltonian VGP.

\subsection{Case 2: At least one diagonal term has unequal single \texorpdfstring{$Z$}{Z} terms}
In this case, once can find an instance for which no $\mathcal{I}$ (single $Z$) terms are zero. This implies that we can divide both sides of Eq.~\eqref{app1:W_re_simp} by the right hand side and obtain
\begin{align}
    &\alpha_{jk}(z_j, \bar{z}_k)\, \alpha_{ik}(z_i, z_k)
    + \alpha_{ij}(\bar{z}_i, \bar{z}_j)\, \alpha_{ik}(z_i, z_k)
    + \alpha_{ij}(\bar{z}_i, \bar{z}_j)\, \alpha_{jk}(z_j, \bar{z}_k)
    = 1\, , \label{app1:alphas_eq}
\end{align}
where we have defined the variables $\alpha_{ij}(x, y) := \frac{\mathcal{R}_{ij}(x, y)}{\mathcal{I}_{ij}(x, y)}$. Then the inequality in~\eqref{app1:W_re_simp} becomes
\begin{align}
    \alpha_{ij}(\bar{z}_i, \bar{z}_j)\, \alpha_{jk}(z_j, \bar{z}_k)\, \alpha_{ik}(z_i, z_k)
    \leq \alpha_{ik}(z_i, z_k)
     + \alpha_{jk}(z_j, \bar{z}_k)
     + \alpha_{ij}(\bar{z}_i, \bar{z}_j)
    \,. \label{app1:alphas_ineq}
\end{align}
As mentioned before, since $\mathcal{I}$ for each diagonal term can take on both positive and negative values, the conditions above must hold for all of them, for a Hamiltonian to be VGP. But if one considers the case where two of the $\alpha$ s namely $\alpha_{ij}$ and $\alpha_{ik}$ are negative and $\alpha_{jk} > 0$, \footnote{that such a condition is must, unless both $h_0$($0$-body term) and $h_3$($2$-body term) are zero, which quite trivially reduces to a contradiction since left hand side of Eq.~\eqref{app1:alphas_ineq} becomes zero, while the right hand size is negative, violating the constraint.} one can divide both sides of the inequality above by the left hand side and obtain the following inequality
\begin{align}
    1
    \leq \frac{1}{\alpha_{jk}(z_j, \bar{z}_k)\, \alpha_{ik}(z_i, z_k)}
    + \frac{1}{\alpha_{ij}(\bar{z}_i, \bar{z}_j)\, \alpha_{ik}(z_i, z_k)}
    + \frac{1}{\alpha_{ij}(\bar{z}_i, \bar{z}_j)\, \alpha_{jk}(z_j, \bar{z}_k)}
    \,. \label{app1:alphas_ineq2}
\end{align}
Since two of the terms in the inequality above and Eq.~\eqref{app1:alphas_eq} will be negative, the two conditions cannot be satisfied. This is because there are no two negative and a positive number that sum to $1$, while the sum of their inverses are greater than $1$.

We conclude this proof by highlighting that this analysis revealed that in order for length-$3$ cycles to satisfy the VGP conditions (i.e. have non-positive real part and vanishing complex part), we must require that at least two of the imaginary parts for the diagonals cancel for all possible spin values.
This was only possible if all diagonal terms have equal single-$Z$ coefficients. \qed

\section{Proof of \texorpdfstring{$\tilde{H}_{\text{Heis.}}$}{H heis} is VGP}
\label{proof:modHheis}
We start by making a couple of key \emph{observations}, stated below:

\textbf{Observation 1.} First, for any configuration of spins on a geometrically frustrated lattice, the number of mis-aligned neighboring spins on any closed loop is \textbf{even}. This is straightforward to realize if one thinks of a closed loop as a 1D spin configuration (with periodic conditions): if all spins align, we have zero mis-alignments, and if we flip one spin, we automatically create two mis-alignments, since the degree of the connections is $two$. Thus, for if we check for alignment pair-wise on two neighboring spins, we would observe \emph{even} number of \emph{mis-alignments}.

\textbf{Observation 2.} Contrary to observation number 1, if we were to check for alignment pair by pair, and then flip the spins that were checked for alignment (which is what $X_i X_j$ do on a given pair of neighboring spins), we should observe that this flips the role of alignment and mis-alignments. In other words, verifying pair-wise alignment on neighboring spins and then flipping the pair of spins will always produce an \emph{even} number of \emph{alignments}, with equally numerous up and down alignments. This is a key observation, and the figure below will help visualize this fact.

Having noted some key observations above, we start by reminding the reader again, that application of a $D_{ij} X_i X_j$ for neighboring spins $i$ and $j$ with spins $z_i = \pm 1$ and $z_j = \pm$ does the following
\begin{align}
    \bra{\bar{z}_i \bar{z}_j} D_{ij} X_i X_j \ket{z_i z_j} = \bra{\bar{z}_i \bar{z}_j}D_{ij} \ket{\bar{z}_i \bar{z}_j} = h^{(ij)}_0 (\bar{z}_i \bar{z}_j - 1) + i h^{(ij)}_1 (\bar{z}_i + \bar{z}_j) \, .
\end{align}
We, hereby, note that if the two neighboring spins mis-align after application of $X_i X_j$, we simply get a factor of $-2h^{(ij)}_0$. The complex phase of this weight is thus $\phi_{\text{mis}} = \frac{\pi}{2} (\text{sgn}(h^{(ij)}_0) + 1)$, i.e. $\phi_{\text{mis.}} = \pi$ if $h^{(ij)}_0 > 0$, and $0$ if $h^{(ij)}_0 < 0$.

However, if the two neighboring spins align, the diagonal factor will simply by $\pm 2 i \, h^{(ij)}_1 $, for the $\pm$ for the two cases where the two spins could be aligned up or down. In this instance, an edge with aligned spins, introduces a completely imaginary factor, with a phase of $\phi_{\text{al.}} = \pm \frac{\pi}{2}$. 

We noted earlier, that checking alignments after flipping will always produce even number of \emph{alignments}, so there will always be even number of $\pm \pi/2$ phases added by the diagonals of a given fundamental cycle. 

For fundamental cycles of \emph{odd} length, we emphasize again that conditions given by Eqs. \eqref{eq:VGP_angle} and \eqref{eq:fc_weight} imply that the phase of the cycle must be $\pi \mod 2\pi$. But, using \textbf{Observation 2} there will always be even number of $\pm \pi/2$ phases in any fundamental cycle, that means that the contribution to the phase from the \emph{alignments} will always be either $0$ or $\pi$ ($\mod 2\pi$). However, if we enforce that $h^{(ij)}_1$ have the same sign for all $D_{ij}$, then the phase contribution from the aligned diagonal weights will simply be zero as any initial phase accrued by an alignment will be canceled by the phase of the diagonal from the next neighboring spins, which are also aligned, but in the opposite direction. 

\begin{figure}[H]
    \centering
    \begin{minipage}[t]{0.33\textwidth}
        \centering
        \includegraphics[width=4.5cm,height=4.5cm]{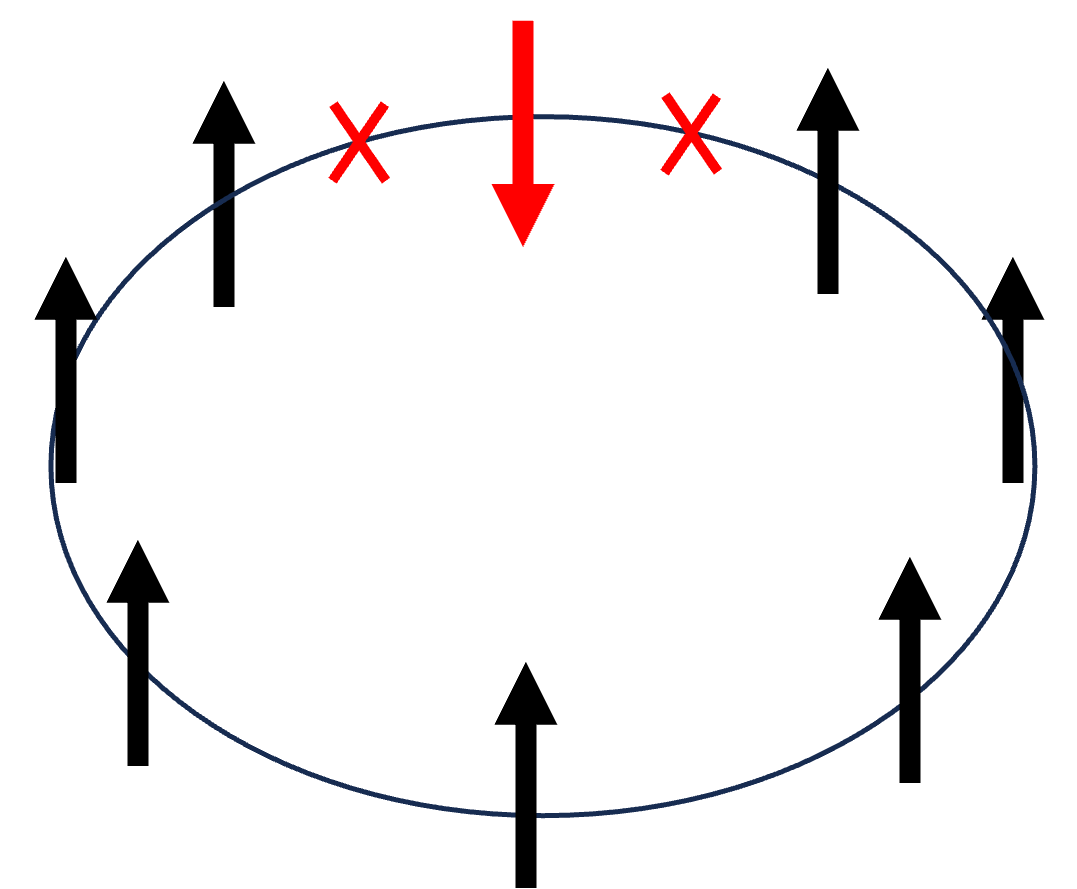}
    \end{minipage}
    \hspace{1cm} 
    \begin{minipage}[t]{0.6\textwidth}
        \centering
        \includegraphics[width=10cm,height=6cm]{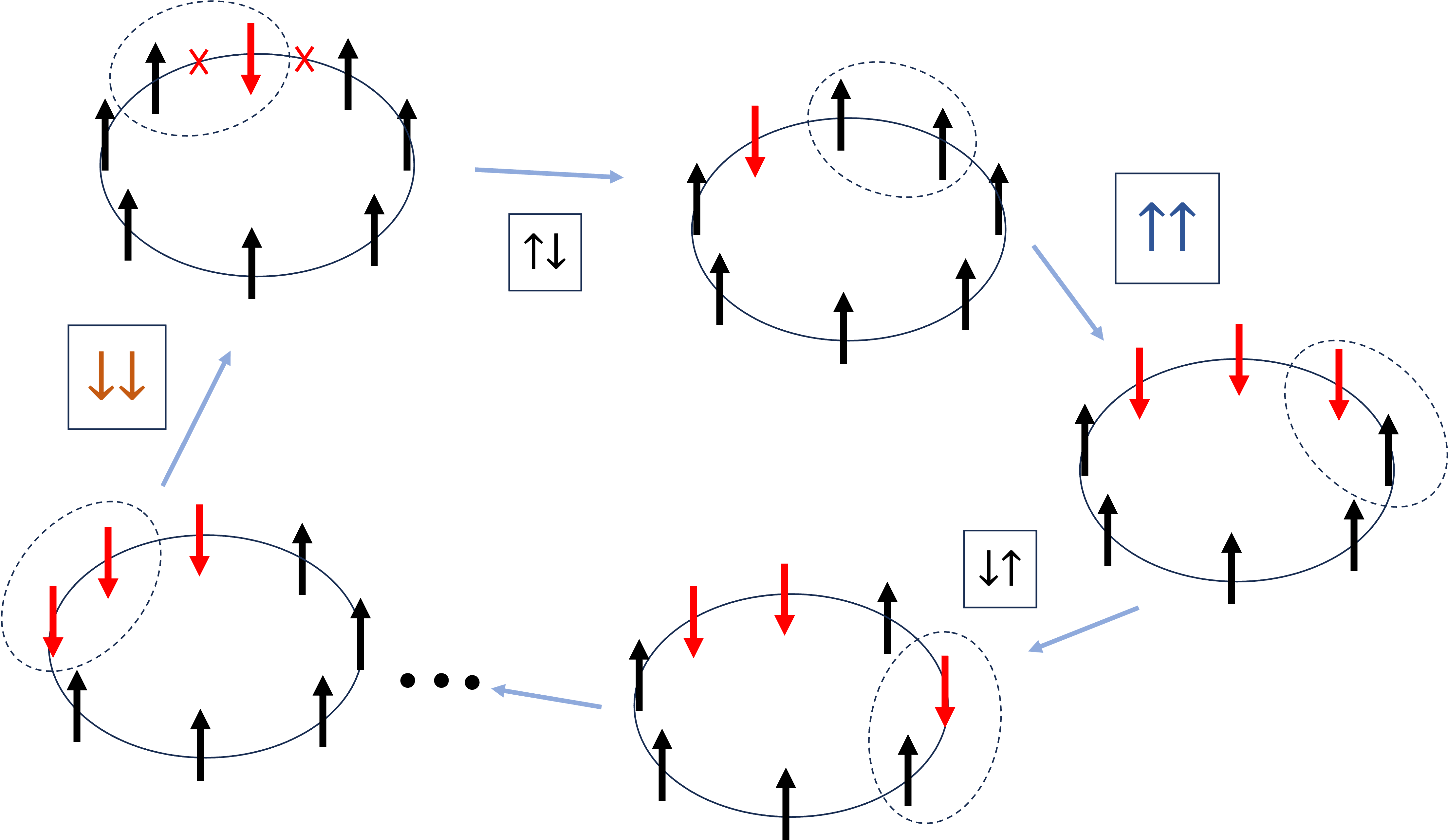}
    \end{minipage}
    \caption{Figure on the left shows an instance of a spin arrangement with even neighboring \emph{misalignments}. On the right, we show how neighboring pair-wise check for alignments, followed by a pair flip, produces even number of opposing \emph{alignments}.}
    \label{fig:side_by_side}
\end{figure}

Finally, one can verify that the fundamental cycles through repeated $D_{ij} X_iX_j$ and $D_{ik} X_iX_k$ (with even powers) are non-negative by checking condition \eqref{eq:D_argcheck}, using the results of Theorem \ref{thm:DX_VGP}. \qed

\section{Heuristical arguments for Conjecture~\ref{conj:2local_periodic_VGP}}
\label{app:2body_periodic_VGP}

If the VGP condition is satisfied on a unit block, then for any spin arrangement, the phases of the different diagonals upon pair-flip and alignment check, as described in Appendix \ref{proof:modHheis}, for a length-$3$ and length-$4$ cycle, will cancel. 
The only concern will be that for larger lattices, one can create a fundamental cycle of larger length than $4$, via a closed fundamental walk across different parts of the lattice. However, since the diagonal operators are the same ones as the unit block, across the entire lattice, any phases accrued will be similar to the ones produced on the unit block lattice. However, we have not formalized a method to prove this rigorously, and so far we have laid this out as a conjecture with numerical support.

\section{Proof of Lemma \ref{lem:key}}
\label{app:lemkey_proof}
First, set $g_X(U) := \mathcal{M}(UXU^\dagger)$. Suppose $U,V \in U(\mathds{D})$ are arbitrary.  Letting $\| \cdot \|_F$ denote the Frobenius norm, we have
\begin{align*}
    \| g_X(V) - g_X(U) \|_F^2 
    &= \sum_{i \neq j} \left|\, |(VXV^\dagger)_{ij}| - |(UXU^\dagger)_{ij}| \,\right|^2 \\&\,\,\,\,\,\,\, + \sum_i \left| (UXU^\dagger)_{ii} - (VXV^\dagger)_{ii} \right|^2 \\
    &\leq \sum_{i,j} \left| (VXV^\dagger)_{ij} - (UXU^\dagger)_{ij} \right|^2 \\
    &= \| VXV^\dagger - UXU^\dagger \|_F^2 \\
    &= \| (V - U)XU^\dagger + VX(V^\dagger - U^\dagger) \|_F^2 \\
    &\leq (\|V - U\|_F \|X\|_{\text{op}} +\\
    & \quad\,\,\|X\|_{\text{op}} \|U - V\|_F )^2 \\
    &= (2\|X\|_{\text{op}} \|U - V\|_F)^2.
\end{align*}
So, the map $g_X$ on $U(N)$ is Lipschitz with constant $L = 2\|H\|_{\text{op}}$ with respect to the Frobenius norm.
\\\\
Next, let $\lambda_k : \text{Herm}(N) \to \mathbb{R}$ denote the $k$th greatest eigenvalue of a given Hermitian matrix. Using this, we expand
\begin{align*}
    \mathcal{L}_X(U) = \ln \left[ \sum_{k=1}^\mathds{D} \exp[\lambda_k(g_X(U))] \right].
\end{align*}
Notice that the LogSumExp function $\text{LSE} : \mathbb{R}^d \to \mathbb{R}$ is $1$-Lipschitz with respect to the sup norm. We show this by letting $u,v \in \mathbb{R}^d$ and setting $w = u - v$. Then
\begin{align*}
    \text{LSE}(u) = \ln\left[ \sum_{i=1}^d e^{v_i + w_i} \right] &= \ln\left[ e^{\|w\|_\infty} \sum_{i=1}^d e^{v_i}e^{w_i - \|w\|_\infty} \right]\\
    &\leq \ln\left[ e^{\|w\|_\infty}\sum_{i=1}^d e^{v_i} \right]\\
    &= \text{LSE}(v) + \|w\|_\infty.
\end{align*}
Here, $\| \cdot \|_{\infty}$ represents the sup norm.
Exchanging the roles of $u$ and $v$ gives the reverse bound, so 
\begin{align*}
    |\text{LSE}(u) - \text{LSE}(v)| \leq \|u-v\|_{\infty}.
\end{align*}
Using Weyl's inequality, this implies
\begin{align*}
    |\mathcal{L}_X(U) - \mathcal{L}_X(V)| &\leq \max_k | \lambda_k(g_X(U)) - \lambda_k(g_X(V)) |\\
    &\leq \| g_X(V) - g_X(U) \|_{\text{op}}\\
    &\leq \| g_X(V) - g_X(U) \|_{\text{F}}\\
    &\leq 2 \|X\|_{\text{op}}\|U - V\|_F
\end{align*}
as desired. \qed

\section{Proof of Theorem \ref{thm:noise_formula}}
\label{app:noiseformula_proof}

    Diagonalize $H = V \Lambda V\dagger$ so that $\Lambda = \operatorname{diag}(\lambda_k(H))$. Write
    \begin{align*}
        (UHU^\dagger)_{ij} = \sum_{k} [\lambda_k(H)] \cdot (UV)_{ik}\overline{(UV)_{jk}}.
    \end{align*}
    Since $V$ is a fixed unitary and $U$ is Haar-random, $UV$ is also Haar-random.
    Take $\mathds{D}$ to be sufficiently large so that the entries of a Haar-random unitary behave as independent, zero-mean complex Gaussians with variance $O(1/\mathds{D})$ \footnote{See the main result of ~\cite{Mastro_2007}, which demonstrates that, asymptotically in the size of a matrix conjugated by a Haar-random unitary, the matrix elements will behave as complex Gaussians}. Then $(UHU^\dagger)_{ij}$ is Gaussian with first two moments $\E[(UHU^\dagger)_{ij}] = 0$ and
    \begin{align*}
        \sigma_{\text{off}}^2 = \E[|(UHU^\dagger)_{ij}|^2] &= \frac{1}{2}\sum_{k} |\lambda_k(UHU^\dagger)|^2 \cdot \E\left[|(UV)_{ik}\overline{(UV)_{jk}}|\right]^2\\ &= \frac{1}{2\mathds{D}^2}\sum_{k} |\lambda_k(UHU^\dagger)|^2\\ &= \frac{\| H \|_F^2}{2\mathds{D}^2}
    \end{align*}
    respectively. Taking the modulus of the Gaussian off-diagonal entries, we immediately get that $H_{ij}' = |(UHU^\dagger)_{ij}|$ follows a Rayleigh distribution with moments $\mu = \E[H_{ij}'] = \sigma_{\text{off}}\sqrt{\frac{\pi}{2}}$ and $\sigma^2 = \text{Var}(H_{ij}') = \left( 2 
- \frac{\pi}{2}\right) \sigma_{\text{off}}^2$~\cite{Siegrist_2019}.
    Likewise, the diagonal entries of $H'$ are approximately 
    identically and  independently distributed (i.i.d.) real Gaussian.

Define the centered and re-scaled matrix
\begin{align*}\label{eq:defW}
W &:=\frac{1}{\sqrt{\mathds{D}}\,\sigma}\bigl(H' - \mu(J-I)\bigr)
\end{align*}
where $J_{ij}=1$ for all $i,j$. Then the entries of $W$ are mean-zero with variance decaying as $O(1/\mathds{D})$, hence $W$ belongs to the Wigner ensemble of random matrices by definition (see Definition 1 of~\cite{Erdos_2009}). That is, its empirical spectral distribution $F^W$ converges, with high probability, to the semicircle distribution~\cite{Gotze_2011}. Next, if we re-arrange
\begin{align*}
    H'=\underbrace{\sqrt{\mathds{D}}\,\sigma\,W-\mu I}_{=:W_{0}}
+\underbrace{\mu J}_{=:P_{1}}
\end{align*}
then we see $W_0$ belongs to the Wigner ensemble (since $-\mu I$ only translates eigenvalues by $O(\mathds{D}^{-1/2})$) and that $P_1$ is a rank-one perturbation. 
By Lemma A.43 of~\cite{Bai_2010}, we see that the Lévy distance $d_L( \cdot\,, \cdot)$ between the empirical spectral distributions $F^{H'}$ and $F^{W_0}$ converges to $0$ as $\mathds{D} \to \infty$:
\begin{equation}
    \label{eq:esd_weak_conv}
    d_L(F^{H'}, F^{W_0}) \leq \frac{\text{rank}(P_1)}{\mathds{D}} = \frac{1}{\mathds{D}}
\end{equation}
where the Lévy distance is a metric on cumulative distribution functions of one-dimensional random variables defined as follows:
Let $F$ and $G$ be cumulative distribution functions (CDFs) on $\mathbb{R}$. 
If $F$ and $G$ are two cumulative distribution functions, the \emph{Lévy distance} between $F$ and $G$ is defined as
\[
L(F,G) \;=\; \inf \Bigl\{ \varepsilon > 0 \;:\; 
F(x - \varepsilon) - \varepsilon \;\leq\; G(x) \;\leq\; F(x + \varepsilon) + \varepsilon 
\quad \text{for all } x \in \mathbb{R} \Bigr\}.
\]
Eq.\eqref{eq:esd_weak_conv} implies that the spectrum $F^{H'}$ weakly converges to $F^{W_0}$ (i.e. it converges to the semicircle law, at least in the bulk of the spectrum). However, an eigenvalue ordering is induced as follows (for details refer to Corollary 4.3.9 of~\cite{Horn_2012})
\begin{align*}
    \lambda_{k+1}(W_0) \leq \lambda_k(H') \leq \lambda_k(W_0) \qquad \text{for $k=\mathds{D},\ldots,2$ }
\end{align*}
implying that at most one eigenvalue of $H'$ may deviate from the spectrum of $W_0$, while the others remain in the semicircle $F^{W_0}$. That is, the leading eigenvalue of $H'$ does not necessarily have to sit between $\lambda_{2}(W_0)$ and $\lambda_{1}(W_0)$.

To determine how far $\lambda_1(H')$ deviates from the bulk, we use the Baik-Ben Arous-Péché (BBP) formula for rank-one additive perturbations of random matrices. $W_0$ is a random matrix and $P_1$ is a rank-one additive perturbation, so we check that the rank–one BBP criterion applies.
The rank–one BBP criterion asserts that whenever $|\mu \mathds{D}|> \sigma\sqrt{\mathds{D}}$ and $\lambda_1(H')$ has a deterministic location (with some small admitted error)\cite{Baik_2005}. Since $\mu \mathds{D}/\sigma\sqrt{\mathds{D}} = \mu\sqrt \mathds{D}/\sigma \to\infty$, the BBP criterion is indeed met and $\lambda_{1}(H')$ satisfies
\begin{equation}
\label{eq:bbp}
\lambda_{1}(H') = \lambda_{\star} + O(\mathds{D}^{-1}),
\end{equation}
where $\lambda_{\star} := -\mu +\mu \mathds{D}+\frac{\sigma^{2}}{\mu}$\cite{Baik_2005}. Since all other eigenvalues almost certainly follow a predictable distribution (i.e. the semicircle law), we can can now write an explicit formula for $\E[\mathcal{F}_H]$. Decompose
\begin{align*}
    \mathcal{F}_H = \underbrace{e^{\lambda_{1}(H')}}_{=:\, \mathcal{F}^{\,\text{out}}_H} + \underbrace{\sum_{k=2}^\mathds{D} e^{\lambda_k(H')}}_{=:\, \mathcal{F}^{{\,\text{bulk}}}_H}.
\end{align*}
In the average case, the empirical spectral measure of $\lambda_2,\ldots,\lambda_k$ is semicircular with density
\begin{align}
    \label{eq:semicircular_density}
    \rho_{sc}(x) = \frac{1}{2\pi\sigma^2\mathds{D}}\sqrt{4\sigma^2\mathds{D}-x^2}  \bf{1} _{\text{$|x|\leq 2\sigma\sqrt{\mathds{D}}$}}.
\end{align}
Here, $\bf{1} _{\text{$|x|\leq 2\sigma\sqrt{\mathds{D}}$}}$ $\in \{0,1\}$ is an indicator function.
Since $U$ was Haar-random, we are in the average case--- so we obtain $\E[\mathcal{F}^{\,\text{bulk}}]$ by integrating ~\eqref{eq:semicircular_density}:
\begin{align*}
    \E[\mathcal{F}^{\,\text{bulk}}] &= (\mathds{D}-1) \int_{-2\sigma\sqrt{\mathds{D}}}^{2\sigma\sqrt{\mathds{D}}} e^x \rho_{sc}(x) \, dx\\ &= (\mathds{D}-1)\frac{I_1(2\sigma\sqrt{\mathds{D}})}{\sigma\sqrt{\mathds{D}}}.
\end{align*}
Along with~\eqref{eq:bbp}, we obtain
\begin{align*}
    \E[\mathcal{F}_H] = e^{\lambda_\star} + (\mathds{D}-1)\frac{I_1(2\sigma\sqrt{\mathds{D}})}{\sigma\sqrt{\mathds{D}}} + \epsilon(H)
\end{align*}
with some error $\epsilon(H)$ satisfying
\begin{align*}
    \epsilon(H) = \underbrace{\left( e^{\lambda_1(H')} - e^{\lambda_\star} \right)}_{:= \epsilon_{\text{out}}} + \underbrace{\frac{1}{\mathds{D}}\left( \sum^\mathds{D}_{k=1} e^{\lambda_k(H')} - e^{\lambda_k(W_0)}\right)}_{:=\epsilon_{\text{bulk}}}.
\end{align*}
Using the upper bound on the Lévy distance in \eqref{eq:esd_weak_conv} together with the
Lipschitz bound on $e^x$ defined over $[-2\sigma\sqrt{\mathds{D}}, 2\sigma\sqrt{\mathds{D}}]$,
\begin{align*}
    \epsilon_{\text{bulk}}
   \leq e^{2\sigma\sqrt{\mathds{D}}}\,
        d_{L}(F^{H'}, F^{W_0}) \leq \frac{e^{2\sigma\sqrt \mathds{D}}}{\mathds{D}}.
\end{align*}
For $\epsilon_\text{out}$, we expand the exponential at the deterministic BBP value
$\lambda_\star$:
\begin{align*}
|\epsilon_{\mathrm{out}}|
   = e^{\lambda_\star}
     \bigl|e^{\lambda_1(H')-\lambda_\star}-1\bigr|
   \leq e^{\lambda_\star}
        |\lambda_1(H')-\lambda_\star|
   \leq e^{\lambda_\star}\frac{C}{\mathds{D}},
\end{align*}
where $C$ is a constant from \eqref{eq:bbp}.
Combining these, we obtain the error bound
\begin{align*}
|\epsilon(H)|
   \leq \frac{e^{2\sigma\sqrt \mathds{D}}+C\,e^{\lambda_\star}}{\mathds{D}} \qquad \implies \qquad\frac{|\epsilon(H)|}{e^{\lambda_\star}}
       =O\bigl(\mathds{D}^{-1}\bigr),
\end{align*}
completing the proof. \qed

\section{Proof of Theorem \ref{thm:exp_f_rand_cliff}}
\label{app:randcliff_proof}
Define the multiset
    \begin{align*}
        S := \{s_1, \ldots, s_M\} \subset \Z_2^N
    \end{align*}
    so that the $k$th entry of each $s_j \in S$ is $1$ if and only if the permutation operator $P_j$ acts non-trivially on spin $k$ (here, we are recording the Pauli $X$-supports of each $\tilde{P}_j$).

    Decompose $S$ into the disjoint union $S = S_1 \,\dot\cup\, S_2$ where $S_1$ contains one representative for each $M - L$ \emph{unique non-zero} member of $S$ and let $S_2$ collect the remaining $L$. For each $P \in S_1$, let $s(P)$ denote its index in $S$. Then define the operator
    \begin{align}
        A := \sum_{k = s(P) \, : \, P \in S_1} |c_{k}| P_{k}.
    \end{align}
    Recall $\mathcal{M}$ is the map that takes the entry-wise absolute value of the off-diagonal entries of an input matrix and negates the sign of each diagonal entry. Then observe that $A = \mathcal{M}(A_{\text{signed}})$ where $A_{\text{signed}} = \sum_{k = s(P) \, : \, P \in S_1} c_{k} \tilde{P}_{k}$. This follows from two key facts:
    \begin{enumerate}
        \item[(i)] $A_\text{signed}$ is a linear combination of Paulis $\tilde{P}_k = D_kP_k$ where $P_k$ is a permutation operator with distinct and non-trivial Pauli $X$-support. That is, no operator in $S_1$ has non-zero diagonal entries and no two permutation operators in $S_1$ overlap on non-zero matrix entries.
        \item[(ii)] $D_k$ is diagonal with entries in $\{\pm 1,\pm i\}$. Hence taking the norm of off-diagonal entries maps $c_kD_kP_k \mapsto |c_k|P_k$.
    \end{enumerate}
    Therefore, studying $\tr(e^A)$ will allow us to study $\mathcal{F}_H(\mathds{1})$ when $|S_2| = L = 0$.
    In order to study the action of $A$, we identify each $v \in \mathbb{C}^{2^N}$ with a linear function $f_v$ defined by $f_v(z) = \langle v | z \rangle$ for each computational basis state $\ket{z}$.
    We may express the action of $A$ in terms of any $f_v$:
    \begin{align}
        f_v(A\ket{z}) = \sum_{k = s(P) \, : \, P \in S_1} |c_k| f_v(z + s_k) = \sum_{k = s(P) \, : \, P \in S_1} |c_k| \bra{v}P_k\ket{z} =\bra{v} A \ket{z}
    \end{align}
    because each permutation $P_k$ translates the computational basis as $P_k\ket{z} = \ket{z+s_k}$. We use this alternative expression of $A$ to compute the eigenvalues of $A$ via eigenfunction decomposition. This will allow us to take the trace:
    \begin{align}
        \tr(e^A) := \sum_{z} \bra{z} e^A \ket{z} = \sum_z f_{z}(e^A \ket{z}) \, .
    \end{align}
    The eigenfunctions of $A$ are given by $\chi_z(x) := (-1)^{z\cdot x}$, for each $z \in \mathbb{Z}_2^N$. To see this, note
     \begin{align}
        \chi_z( A \, x) &= \sum_{j=1}^M |c_j|\chi_z(x+s_j)\notag \\
        &= \sum_{j=1}^M |c_j| (-1)^{z\cdot(x+s_j)}\notag \\
        &= \underbrace{\left(\sum_{j=1}^M |c_j| (-1)^{z\cdot s_j}\right)}_{:= \lambda_z} \chi_z(x).\notag
    \end{align}
\noindent
Let us, now, define
\begin{align*}
\chi^{(k)}_z:=(-1)^{z\cdot s_k}\in\{\pm1\},
\end{align*}
where $s_k \in S$ and $z$ is fixed. We establish the following identities for $z \neq \vec{0}$ on the expectation of moments of $\chi^{(k)}_z$, noting that the case for $z = \vec{0}$ is quite trivial as explained later,
\begin{align*}
\mathbb{E}_{s_k}[\chi^{(k)}_z]
  &= \frac{1}{2^{N}-1}\sum_{s\in\mathbb{Z}_2^{N}\setminus\{0\}}(-1)^{z\cdot s}
     = \frac{-1}{2^{N}-1},\notag \\
\mathrm{Var}(\chi^{(k)}_z)&=\mathbb{E}[(\chi^{(k)}_z)^{2}]-\left(\mathbb{E}[\chi^{(k)}_z]\right)^{2}
     = 1-\frac{1}{(2^{N}-1)^{2}},\notag \\
\mathbb{E}_{s_k , s_{\ell}}[\chi^{(k)}_z\chi^{(\ell)}_z]
  &=\frac{1}{(2^{N}-1)(2^{N}-2)}
    \sum_{\substack{s\neq t\\s,t\neq 0}}(-1)^{z\cdot(s+t)}
    =\frac{-1}{2^{N}-1}\quad(k\neq \ell),\notag \\
\left|\mathbb{E}_{s_k , s_\ell , s_m}[\chi^{(k)}_z\chi^{(\ell)}_z\chi^{(m)}_z]\right|
  &\le\frac{1}{(2^{N}-1)(2^{N}-2)(2^{N}-3)}
        \sum_{\substack{s\neq t\neq r\\s,t,r\neq 0}}1
    =O\left(2^{-2N}\right)\quad(\text{$k,\ell,m$ all distinct}).
\end{align*}

We would like to compute the expectation value of the exponential of the eigenvalues of $A$. Using \(e^{aX}=\cosh a+X\sinh a\) for any \(X\) with \(X^{2}=1\), we have
\begin{align}
\chi_z\Big(e^{A} \ket{z}\Big)
=\prod_{ k = s(P) \, : \, P \in S_1}\bigl(\cosh(|c_k|)+(-1)^{z \cdot s_k}\sinh(|c_k|)\bigr).
\end{align}

Recall $L$ is the number of duplicates and non-zero binary strings we have in the multiset $S$ (that is, $|S_2| = L$). We condition on $L = 0$ and compute the conditional expectation $\E[\mathcal{F}_H \, \vert \, L = 0]$ by finding $\E[e^{\lambda_z} \, \vert \, L = 0]$ for each computational basis state $z$.
We can expand the product, and take the expectations term by term, keeping contributions up to second order, as we have shown above that third order corrections are small:
\begin{align*}
\mathbb{E} \bigl[e^{\lambda_z} \, \vert \, L = 0\bigr]
&=
\prod_{ k=1}^M\cosh(|c_k|)
\Bigg[
1
+\sum_{k= 1}^M\tanh(|c_k|)\,\mathbb{E}_{s_k}[\chi^{(k)}_z]\\
&\qquad\qquad\qquad
+\sum_{k < \ell}
\tanh(|c_k|)\tanh(|c_\ell|)\,\mathbb{E}_{s_k , s_\ell}[\chi^{(k)}_z\chi^{(\ell)}_z]
+O\!\Bigl(\tfrac{M^{3}}{2^{2N}}\Bigr)
\Bigg]\\
&=
\prod_{ k=1}^M\cosh(|c_k|)
\Bigg[
1
-\frac{\displaystyle\sum_{k=1}^M\tanh(|c_k|)}{2^{N}-1}
-\frac{\displaystyle\sum_{k < \ell}
\tanh(|c_k|)\tanh(|c_\ell|)}{2^{N}-1}
+O\!\Bigl(\tfrac{M^{3}}{2^{2N}}\Bigr)
\Bigg].
\end{align*}
\noindent
Using $0\le \tanh(|c_k|)\le 1$ and $\binom{|S_1|}{2}\le \frac{M^2}{2}$,
\begin{align*}
\left|
-\frac{\sum_{k=1}^M\tanh(|c_k|)}{2^{N}-1}
-\frac{\sum_{k < \ell}\tanh(|c_k|)\tanh(|c_\ell|)}{2^{N}-1}
\right|
\le
\frac{M+\binom{M}{2}}{2^{N}-1}
=O\!\bigl(M^{2}2^{-N}\bigr).
\end{align*}
Lastly, notice that the value of $\lambda_{\vec{0}}$ is deterministic as long as $L=0$:
\begin{align*}
    \lambda_{\vec{0}} = \sum_{k=1}^M |c_k| (-1)^{0 \cdot s_k} =: c_{\text{tot}}.
\end{align*}
Altogether,
\begin{align}
\mathbb{E} \left[\sum_{z\in\mathbb{Z}_2^{N}}e^{\lambda_z} \, \Bigg\vert \, L = 0  \right]
&=
e^{c_{\mathrm{tot}}}
+(\mathds{D}-1)
\prod_{k=1}^M\cosh(|c_k|)
\Bigg[
1
-\frac{\displaystyle\sum_{k=1}^M\tanh(|c_k|)}{\mathds{D}-1}
\quad-\frac{\displaystyle\sum_{k<\ell}
\tanh(|c_k|)\tanh(|c_\ell|)}{\mathds{D}-1}
+O\!\Bigl(\tfrac{M^{3}}{2^{2N}}\Bigr)
\Bigg]\notag
\\
\label{eq:exp_F_H_given_L_0}
&=
e^{c_{\mathrm{tot}}}
+(\mathds{D}-1)
\left[\prod_{k=1}^M\cosh(|c_k|)\right]
\left[1+O\!\bigl(M^{2}2^{-N}\bigr)\right].
\end{align}
\noindent
Finally, to determine the probability that $H$ is in the class of Hamiltonians with $L = 0$, we bound
\begin{align*}
    \Pr[L > 0] \leq \sum_{k \neq \ell} \Pr[s_k = s_\ell] + \sum_{k} \Pr[s_k = 0] \leq \binom{M}{2} \frac{1}{\mathds{D}-1} + \frac{M}{\mathds{D}} = O(M^22^{-N}).
\end{align*}
Then with probability $1 - O(M^22^{-N})$, a spin-1/2 Hamiltonian composed of $M =O(\mathrm{poly}(N))$ random Pauli operators belongs to a class of Hamiltonians $\mathcal{H}$ with expected value of $\mathcal{F}_H(\mathds{1}) \vert_\mathcal{H}$ given by Eq.~\eqref{eq:exp_F_H_given_L_0}. This completes the proof.  \qed
\end{document}